%% file: Mean.tex
\documentclass[USenglish,a4paper,11pt]{article}
\pdfoutput=1

\input{Style.sty}
\title{Near-Optimal Quantum Algorithms for \\ Multivariate Mean Estimation}

\author{Arjan Cornelissen\thanks{QuSoft, University of Amsterdam. \url{arjan.cornelissen@cwi.nl}}
   \and Yassine Hamoudi\thanks{Simons Institute for the Theory of Computing, University of California, Berkeley. \url{hamoudi@berkeley.edu}}
   \and Sofiene Jerbi\thanks{Institute for Theoretical Physics, University of Innsbruck. \url{sofiene.jerbi@uibk.ac.at}}}

\date{\today}

\begin{document}

\maketitle

\begin{abstract}
  \input{Abstract.tex}
\end{abstract}

\section{Introduction}
\input{Introduction.tex}

\section{Preliminaries}

  \subsection{Notations}
  \input{Notations.tex}

  \subsection{Input models}
  \input{Models.tex}

  \subsection{Algorithmic tools}
  \input{Tools.tex}

\section{Mean estimation with binary oracles}
\label{Sec:MeanBinary}

  \subsection{Bounded multivariate estimator}
  \label{Sec:BoundedBinary}
  \input{BoundedBinary.tex}

  \subsection{Near-optimal multivariate estimator}
  \label{Sec:MultiBinary}
  \input{MultiBinary.tex}

  \subsection{Lower bounds}
  \label{Sec:LowerBinary}
  \input{LowerBinary.tex}

\section{Mean estimation with phase oracles}
\label{Sec:Analog}
\input{AnalogEstimator.tex}

  \subsection{Lower bounds}
  \label{Sec:LowerAnalog}
  \input{LowerAnalog.tex}

\section{Applications}
\label{Sec:Applications}
\input{Applications.tex}

\section{Discussion}
\label{Sec:Discussion}
\input{Discussion.tex}

\section*{Acknowledgments}
\input{Acknowledgements.tex}

\printbibliography[heading=bibintoc]


\end{document}

%% file: Abstract.tex
We propose the first near-optimal quantum algorithm for estimating in Euclidean norm the mean of a vector-valued random variable with finite mean and covariance. Our result aims at extending the theory of multivariate sub-Gaussian estimators~\cite{LM19j} to the quantum setting. Unlike classically, where any univariate estimator can be turned into a multivariate estimator with at most a logarithmic overhead in the dimension, no similar result can be proved in the quantum setting. Indeed, Heinrich~\cite{Hei04j} ruled out the existence of a quantum advantage for the mean estimation problem when the sample complexity is smaller than the dimension. Our main result is to show that, outside this low-precision regime, there does exist a quantum estimator that outperforms any classical estimator. More precisely, we prove that the approximation error can be decreased by a factor of about the square root of the ratio between the dimension and the sample complexity. Our approach is substantially more involved than in the univariate setting, where most quantum estimators rely only on phase estimation. We exploit a variety of additional algorithmic techniques such as linear amplitude amplification, the Bernstein-Vazirani algorithm, and quantum singular value transformation. Our analysis is also deeply rooted in proving concentration inequalities for multivariate truncated statistics.

We develop our quantum estimators in two different input models that showed up in the literature before. The first one provides coherent access to the binary representation of the random variable and it encompasses the classical setting. In the second model, the random variable is directly encoded into the phases of quantum registers. This model arises naturally in many quantum algorithms but it is often incomparable to having classical samples. We adapt our techniques to these two settings and we show that the second model is strictly weaker than the other one for solving the mean estimation problem. Finally, we describe several applications of our algorithms, notably in measuring the expectation values of commuting observables and in the field of machine learning.

%% file: Introduction.tex
Monte Carlo methods are used extensively in various fields of science and engineering, such as statistical physics~\cite{BH10b}, finance~\cite{Gla03b}, or machine learning~\cite{AdFDJ03j}. At the core of these methods is a Monte Carlo process, e.g., a randomized algorithm, whose \emph{expected outcome} is to be estimated via repeated random executions. Quantum computers can speed-up this approach at two different levels~\cite{Mon15j}. First, novel algorithmic techniques such as Hamiltonian simulation~\cite{Fey82j} or quantum walks~\cite{Sze04c} provide faster Monte Carlo simulation processes. Secondly, quantum metrology algorithms (such as phase estimation~\cite{Kit95p}) give better error rates for computing statistics on these processes. The present paper focuses on this second point through the lens of the \emph{mean estimation problem}. In this problem, the objective is to compute the closest possible \emph{estimate}~$\mut$ to the mean $\mu = \ex{X}$ of a random variable~$X$ representing the output of some black-box process. Given the ability to repeat this process~$n$ times (the \emph{sample complexity}), one seeks to minimize the error $\norm{\mut - \mu}$ made with high probability.

In the classical setting, a beautiful theory~\cite{LM19j} has been developed to solve the mean estimation problem in Euclidean norm. Under the sole assumption that the covariance matrix~$\Sigma$ of~$X$ exists, it turns out that the optimal \emph{non-asymptotic} error behaves as if~$X$ followed the Gaussian distribution $\mathcal{N}(\mu,\Sigma)$. This motivated the use of the adjective \emph{sub-Gaussian} to qualify the optimal classical estimators. In one dimension, the most well-known sub-Gaussian estimator is arguably the median-of-means~\cite{NY83b,JVV86j,AMS99j}. The first computationally efficient sub-Gaussian estimator in high dimension was only found recently by Hopkins~\cite{Hop20j}. These estimators achieve an optimal error of $\norm{\mut - \mu}_2 \leq \bo[\big]{\sqrt{\Tr(\Sigma)/n} + \sqrt{\norm{\Sigma}\log(1/\delta)/n}}$ with probability~$1-\delta$.

In the quantum setting, the univariate case $X \in \R$ has been studied since the early works on quantum counting~\cite{BBHT98j}. The celebrated amplitude estimation algorithm~\cite{BHMT02j} provides a smaller error rate for estimating the mean of any \emph{Bernoulli} random variable compared to the classical estimators. For general univariate distributions, a series of quantum estimators~\cite{Gro98c,Ter99d,AW99p,Hei02j,WCNA09j,BDGT11p,Mon15j,HM19c,Ham21c} culminated into a near-optimal algorithm that outperforms any classical estimator. On the other hand, the multivariate case $X \in \R^d$, appearing notably in machine learning applications, remains largely unaddressed by quantum algorithms. Classically, it admits a simple near-optimal approach: the $d$ coordinates of $\mu$ can all be estimated simultaneously with $d$ univariate sub-Gaussian estimators run in parallel (i.e., using the same samples from $X$) with only a logarithmic overhead $\log(d)$ in sample complexity (due to the Hoeffding bound). In the quantum scenario however, this \emph{simultaneous} evaluation of several univariate expectation values is more complicated. Indeed, the quantum algorithms for the univariate case rely on quantum amplitude estimation \cite{BHMT02j}, which involves as a critical step an encoding of the expectation value in the relative phase of a quantum register. At first sight, it is unclear how a vector of~$d$ phases could be encoded simultaneously into~$d$ registers without requiring a linear overhead in~$d$. In fact, a lower bound proved by Heinrich~\cite{Hei04j} rules out the possibility of simply a $\log(d)$ overhead for the quantum multivariate mean estimation problem.

Our paper develops \emph{near-optimal} and \emph{computationally efficient} quantum mean estimators for vector-valued random variables of arbitrary dimension with \emph{binary oracle} access. Unlike in the univariate setting ($d=1$), where the optimal quantum estimator~\cite{Ham21c} is strictly more efficient than any classical estimator, we identify two different regimes in higher dimension: (i) if a quantum estimator is limited to accessing the input at most $d$ times (i.e. $n \leq d$) then no advantage can be gained over the classical sub-Gaussian estimators, (ii) if it can access the input at least $d$ times (i.e. $n \geq d$) then the approximation error can be reduced by a near-optimal factor of $\sqrt{d/n}$ compared to classical sub-Gaussian estimators.

We complement this work with new quantum estimators in the weaker \emph{phase oracle} access model, where the information about~$X$ are directly encoded into the phases of quantum registers. This model has been considered before~\cite{GAW19c}, albeit not in the context of quantum mean estimation. We adapt some of our techniques to this model and show that here we can even obtain near-optimal estimators with respect to any $\ell_p$-norm, with $p \in [1,\infty]$, thereby providing a complete characterization of the query complexities involved in the mean estimation problem. This part of our work shares some overlap with a related paper by a subset of the authors~\cite{CJ21p} that focused on the probability and phases oracles models for multivariate Monte Carlo estimation.


\subsection{Contributions}

Our main contribution is the design of new quantum mean estimators that achieve the best possible error rates, up to logarithmic factors, in the multivariate setting. We investigate this problem in two different quantum input models. We first consider the \emph{binary oracle} model in Section~\ref{Sec:MeanBinary}, which generalizes in a natural way the classical \emph{sample complexity} and is the most frequent setting used in previous work (e.g.~\cite{AW99p,Hei02j,BHH11j,BDGT11p,Mon15j,Ham21c}). In this model, the access to a $d$-dimensional random variable $X : \Omega \ra \R^d$ over a probability space $(\Omega, 2^{\Omega}, \P)$ is provided through two unitary operators: one that prepares a superposition over the probability space $U_{\P} : \ket{0} \mapsto \sum_{\omega \in \Omega} \sqrt{\P(\omega)} \ket{\omega}$, and one that evaluates the random variable over the sample set $\bora_X : \ket{\omega}\ket{\vec{0}} \mapsto \ket{\omega}\ket{X(\omega)}$. Note that the mean to be estimated is $\mu = \sum_{\omega \in \Omega} \P(\omega) X(\omega)$. Our first main contribution is to provide an optimal multivariate quantum mean estimator in this setting. Our approach is substantially more involved than in the univariate case~\cite{Ham21c}. A core primitive developed in our work is a new estimator for random variables bounded in $\ell_2$-norm. This can be seen as a multivariate version of the well-known Amplitude Estimation algorithm~\cite{BHMT02j}. Our techniques are based on the Bernstein-Vazirani algorithm~\cite{BV97j} (more precisely, its generalization to estimating a linear function over the reals~\cite{Jor05j}), the quantum singular value transformation framework~\cite{GSLW19c}, and tail inequalities for truncated statistics. We state our first result below with respect to the $\ell_{\infty}$-norm to highlight that it will be more natural to use this metric in our algorithms and convert to $\ell_2$-norm by standard norm inequalities.

\begin{rtheorem}[Theorem~\ref{Thm:MultiBoundedEst} {\normalfont (Informal)}]
  There is a quantum estimator that estimates the mean $\mu$ of any $d$-dimensional random variable $X$ with error $\norm{\mut - \mu}_\infty \leq \frac{\sqrt{L_2} \log(d)}{n}$ and success probability $2/3$, given an upper bound $L_2 \geq \ex{\norm{X}_2}$ and $\wbo{n}$ queries to the oracles $U_{\P}$ and $\bora_X$. The error made by this estimator in $\ell_2$-norm is $\norm{\mut - \mu}_2 \leq \frac{\sqrt{d L_2} \log(d)}{n}$.
\end{rtheorem}

As an illustration of this result, one can simultaneously estimate the expectation values of~$d$ univariate random variables $X_1, \dots, X_d$ distributed in $[0,1]$ each with error $\sqrt{d}\log(d)/n$ by doing~$\wbo{n}$ queries. In comparison, running the Amplitude Estimation algorithm on each random variable separately (with $\wbo{n/d}$ queries) would result in an error of $d/n$.

Similarly to the Amplitude Estimation algorithm, the above primitive estimator does not always provide an optimal error rate with respect to the trace $\Tr(\Sigma) = \ex{\norm{X}_2^2} - \norm{\ex{X}}_2^2$ of the covariance matrix~$\Sigma$ of~$X$. Moreover, it requires the $\ell_2$-norm of $X$ to be bounded by $1$. We improve upon this result to design the next optimal quantum mean estimator.

\begin{rtheorem}[Theorem~\ref{Thm:MultiEuclidean} {\normalfont (Informal)}]
  There is a quantum estimator in the binary oracle model that estimates the mean $\mu$ of any $d$-dimensional random variable $X$ with error
  	\[\norm{\widetilde{\mu} - \mu}_2 \leq
        \begin{cases}
      		\sqrt{\frac{\Tr(\Sigma)}{n}}, & \text{if $n \leq d$}, \\[2mm]
      		\frac{\sqrt{d \Tr(\Sigma)} \log(d)}{n}, & \text{if $n > d$,}
      	\end{cases}\]
  and success probability $2/3$, given $\wbo{n}$ queries to the oracles $U_{\P}$ and~$\bora_X$.
\end{rtheorem}

This bound is achieved by using any classical sub-Gaussian estimators~\cite{LM19j} when $n \leq d$, and a new quantum estimator when $n \geq d$. We show that these two regimes are inevitable since no quantum speed-up is possible when $n \leq d$, whereas our quantum estimator is optimal when $n \geq d$. Our lower bounds are based on the quantum query complexity of approximating a bit string whose entries are determined by parity functions.

\begin{rtheorem}[Theorems~\ref{Thm:LBLowBinary} and~\ref{Thm:LBHighBinary} {\normalfont (Informal)}]
  For any estimator that uses at most $n$ binary oracle queries, there is a $d$-dimensional random variable~$X$ such that, with probability~$2/3$, the error is at least $\norm{\mut - \mu}_2 \geq \om[\Big]{\sqrt{\frac{\Tr(\Sigma)}{n}}}$ if $n \leq d$, and $\norm{\mut - \mu}_2 \geq \om[\Big]{\frac{\sqrt{d\Tr(\Sigma)}}{n}}$ if $n \geq d$.
\end{rtheorem}

Next, in Section~\ref{Sec:Analog}, we investigate the mean estimation problem in the \emph{phase oracle} model where the aforementioned unitary $\bora_X$ is replaced with phase access $\pora_X : \ket{\omega}\ket{j} \mapsto e^{i X(\omega)_j} \ket{\omega}\ket{j}$ to the coordinates of $X$. This model can be efficiently simulated using a binary oracle but the converse is generally not true. In fact, even obtaining one classical sample from $X$ using a phase oracle is generally a hard task. On the other hand, as explained in \cite{GAW19c}, this model arises naturally in the context of variational quantum eigensolvers, QAOA, and quantum auto-encoders for instance. This reason motivates understanding what is the optimal error rate for mean estimation in this weaker setting. Although a phase oracle does not allow to obtain an error depending on the covariance matrix, we manage to adapt some of the techniques developed for binary oracles to arrive at an optimal estimator when~$X$ is bounded in $\ell_{\infty}$-norm by $\norm{X}_{\infty} \leq 1/4$. Interestingly, our results differ qualitatively from those in the binary oracle setting in two aspects. First, our estimator does not make the same number of queries to the oracles $U_{\P}$ and~$\pora_X$, and the optimal precision depends in fact differently on these two parameters. Second, in this model, we are actually able to tightly characterize the optimal performance with respect to all $\ell_p$-norms, where $p \in [1,\infty]$, up to polylogarithmic factors. We state our results here only with respect to the $\ell_2$-norm for ease of exposition and comparison to the binary oracle setting, and we show that these results are nearly optimal in Theorem~\ref{Thm:LowerPO}.

\begin{rtheorem}[Theorem~\ref{Thm:phOracle} {\normalfont (Informal)}]
  There is a quantum estimator that estimates the mean $\mu$ of any $d$-dimensional random variable $X$ such that $\norm{X}_{\infty} \leq 1/4$ with error
  	\[\norm{\widetilde{\mu} - \mu}_2 \leq
        \begin{cases}
      		\max\set[\Big]{\sqrt{\frac{d}{n}}, \frac{d^{3/2}}{n'}} \cdot \log(d), & \text{if $n \leq d$}, \\[2mm]
      		\max\set[\Big]{\frac{d}{n}, \frac{d^{3/2}}{n'}} \cdot \log(d), & \text{if $n > d$,}
      	\end{cases}\]
 and success probability $2/3$, given $\wbo{n}$ queries to the oracle $U_{\P}$ and~$\wbo{n'}$ queries to the oracle~$\pora_X$.
\end{rtheorem}

Finally, we conclude this paper by giving some applications of the above results in Section~\ref{Sec:Applications}. We first explain how our formulation of the multivariate mean estimation problem covers the general task of estimating the expectation values of several mutually commuting observables with respect to a given quantum state (Section~\ref{Sec:comObs}). We then present several applications in the literature, and notably in quantum machine learning (training variational quantum circuits, Boltzmann machines, or reinforcement learning policies), where this problem arises (Section~\ref{Sec:ExApps}).


\subsection{Proof overview}
\label{Sec:Overview}

We give a high-level description of the algorithms developed in Section~\ref{Sec:MeanBinary} for addressing the mean estimation problem in the binary oracle model. Similar techniques are employed in Section~\ref{Sec:Analog} for the phase oracle model. We simplify the exposition by replacing $\Tr(\Sigma) = \ex{\norm{X - \mu}_2^2}$ with the second moment~$\ex{\norm{X}_2^2}$, and by taking the failure probability~$\delta$ to be a small constant. The approximation error $\norm{\mut - \mu}_{\infty}$ is measured here with respect to the $\ell_{\infty}$-norm.

\paragraph*{Bounded multivariate estimator.}
The main obstacle when trying to generalize most quantum univariate estimators (e.g.~\cite{Ter99d,Hei02j,WCNA09j,Mon15j,HM19c,Ham21c}) to the multivariate setting is the absence of an estimator for \emph{bounded} multivariate random variables. In the univariate setting, such an estimator is provided by the well-known Amplitude Estimation algorithm~\cite{BHMT02j} which, by a well-known trick~\cite{Ter99d,WCNA09j,Mon15j}, can estimate the mean of any random variable \emph{bounded in $[-1,1]$} with an error on the order of $\sqrt{\ex{\abs{X}}}/n$. It is worth recalling how this estimator works when $X$ is bounded in $[0,1]$: the value $\varphi = \arcsin(\sqrt{\ex{X}})$ is encoded as the phase of a particular unitary operator and estimated with error $1/n$ using phase estimation~\cite{Kit95p}. Then, by standard trigonometric identities, $\abs{\td{\varphi} - \varphi} \leq 1/n$ implies that $\abs{\sin^2(\td{\varphi})-\ex{X}} \leq 2\sqrt{\ex{X}}/n + 1/n^2$ (the lower-order term $1/n^2$ can be removed by  testing if $\ex{X} \leq 1/n^2$ and outputting~$0$ if this is the case). We generalize this idea to higher dimensions in a novel way by considering the \emph{directional mean} function $u \mapsto \inp{u}{\ex{X}}$ where $u \in \R^d$. By using a constant number of queries to $X$ and amplitude-to-phase conversion techniques~\cite{GAW19c}, one can efficiently approximate the unitary $\ket{u} \mapsto e^{i \inp{u}{\ex{X}}} \ket{u}$ if $\abs{\inp{u}{X}} \leq 1$ almost surely. We could then estimate the directional mean $\inp{u}{\ex{X}}$ with phase estimation, for sufficiently many values of $u$, in order to reconstruct an estimate of $\ex{X}$. However, this approach would incur a linear cost in the dimension~$d$. Instead, since the directional mean is a \emph{linear function} in $u$, we can use a variant of the Bernstein-Vazirani algorithm~\cite{BV97j} to directly recover the entire vector $\ex{X}$ (up to a certain precision) with fewer queries. This idea is also at the heart of the quantum gradient estimation algorithms~\cite{Jor05j,GAW19c}, however it requires two major improvements for our setting. First, we can only make the assumptions that $X$ is bounded in $\ell_2$-norm (i.e. $\norm{X}_2 \leq 1$) and $u$ in $\ell_{\infty}$-norm (i.e. $\norm{u}_{\infty} \leq 1$). However, these two conditions do not imply that $\abs{\inp{u}{X}} \leq 1$ as needed by the amplitude-to-phase conversion technique. We overcome this issue by proving tail inequalities for inner products and directional means (Lemma~\ref{Lem:Truncate}) showing that they do not exceed~$1$ with high probability under our assumptions. Hence, by suitable truncations, this gives us a first version of a bounded estimator with error~$1/n$. Secondly, we need to incorporate information about $X$ in the error. We cannot reproduce the univariate approach by encoding $\arcsin(\sqrt{\abs{\inp{u}{\ex{X}}}})$ instead of $\inp{u}{\ex{X}}$ into the phase, since it would no longer be a linear function. Instead, we use the quantum singular value transformation framework~\cite{GSLW19c} to linearly amplify the (squared) amplitude encoding the directional mean $\inp{u}{\ex{X}}$ into $\inp[\big]{u}{\frac{\ex{X}}{\ex{\norm{X}_2}}}$, before applying the amplitude-to-phase conversion technique. Since this amplification step requires $\bo{1/\sqrt{\ex{\norm{X}_2}}}$ queries, this leaves us with $\wbo{n \sqrt{\ex{\norm{X}_2}}}$ iterations available for the vector recovering step. Hence, the rescaled mean~$\frac{\ex{X}}{\ex{\norm{X}_2}}$ is estimated with error~$1/(n \sqrt{\ex{\norm{X}_2}})$, which translates into the improved error of $\sqrt{\ex{\norm{X}_2}}/n$ for $\ex{X}$ (Theorem~\ref{Thm:MultiBoundedEst}).

\paragraph*{Near-Optimal multivariate estimator.}
We build on the above bounded estimator to remove the assumption on the boundedness of~$X$ and decrease the error to $\sqrt{\ex{\norm{X}_2^2}}/n$. Similarly to the univariate case~\cite{Hei02j,Mon15j,HM19c,Ham21c}, we decompose $X = X_0 + X_1 + X_2 + \dots$ into a sequence of \emph{truncated} random variables $X_{j} = X \ind{a_{j-1} < \norm{X}_2 \leq a_{j}}$ over slices of the $\ell_2$-ball, where the values outside the range $(a_{j-1},a_{j}]$ are mapped to $0$. The truncation levels $0 = a_{-1} < a_0 < a_1 < a_2 < \dots$ are chosen so that the bounded estimator performs well on each~$X_{j}$ individually. In the univariate setting, this sequence followed a geometric progression of ratio $2$. Here, we instead choose $a_{j}$ to be the \emph{quantile value} of order $2^{-j}$ satisfying $\pr{\norm{X}_2 \geq a_{j}} \approx 2^{-j}$. This new choice has the advantage that the expected norm of $X_{j}$ can be \emph{explicitly} bounded as $\ex{\norm{X_{j}}_2} \leq 2^{-j-1}$ (Equation~(\ref{Eq:MomBound})), a property needed by our bounded estimator. Moreover, we show that this sequence increases slowly enough so that $a_{j} \leq 2^{j/2} \sqrt{\ex{\norm{X}_2^2}}$ (Equation~(\ref{Eq:quantBound})). Consequently, the bounded estimator can estimate \emph{separately} the mean of each $X_{j}$ with an error of $a_{j} \sqrt{\ex{\norm{X_{j}}_2}}/n \leq 2^{-1/2} \sqrt{\ex{\norm{X}^2_2}}/n$ (where the $a_{j}$ factor comes from normalizing~$X_{j}$ to make it fit into the unit $\ell_2$-ball). Finally, each quantile $a_{j}$ can be computed (approximately) in time $\bo{2^{-j/2}}$ using the quantile estimator developed in~\cite{Ham21c} (Proposition~\ref{Prop:Quantile}), and we only need to consider $j \leq \bo{\log n}$ truncated random variables since the part of $X$ above that threshold does not contribute to a significant portion of the mean (Equation~(\ref{Eq:truncLargest})). This leads to the final error of $\sqrt{\ex{\norm{X}^2_2}}/n$ (Theorem~\ref{Thm:MultiEstim}).


\subsection{Related work}

There is an extensive literature on classical mean estimators and we refer the reader to \cite{LM19j} for an excellent survey on the optimal \emph{sub-Gaussian estimators} in Euclidean norm. We point out that the \emph{empirical mean} estimator is generally not optimal, and its error is captured by several standard concentration bounds such as the Chebyshev, Chernoff and Bernstein inequalities.

There is a series of quantum \emph{univariate} mean estimators \cite{Gro98c,AW99p,BDGT11p} that get close to the error $1/n$ for random variables distributed in $[0,1]$ (and success probability~$2/3$). The amplitude estimation algorithm~\cite{BHMT02j,Ter99d} leads to a sharper bound of $\sqrt{\mu}/n$. Nevertheless, the mean $\mu$ is always larger than or equal to the variance~$\sigma^2$ when~$X$ is distributed in $[0,1]$. The question of improving the dependence on~$\sigma^2$ was considered in~\cite{Hei02j,Mon15j,HM19c,Ham21c}, where it is shown that the optimal error is~$\sigma/n$.

There are very few works addressing the quantum \emph{multivariate} mean estimation problem. Heinrich~\cite{Hei04j} proved that the error rate must depend on $1/n$ when the dimension is sufficiently large. Our lower bound in the $n \leq d$ regime (Theorem~\ref{Thm:LBLowBinary}) refines this statement by adding a dependence on the covariance matrix. In a recent work, van Apeldoorn~\cite{vApe21c} proposed a ``multidimensional Amplitude Estimation'' algorithm. However, it only applies to a restricted set of random variables and the error does not recover that of Amplitude Estimation when $d = 1$. More precisely, the author described a quantum algorithm for estimating with error $1/n$ (in $\ell_{\infty}$-norm) a probability vector $p = (p_1,\dots,p_d)$ given access to a unitary $U : \ket{0} \mapsto \sum_i \sqrt{p_i} \ket{i}$. This is a special case of the multivariate mean estimation problem, where the random variable $X \in \rn^d$ is equal to the basis vector~$e_i$ with probability~$p_i$. Applying our main result (Theorem~\ref{Thm:MultiEstim}) to $X$ decreases the error given in~\cite{vApe21c} by a factor of $\sqrt{\Tr(\Sigma)} = \pt{1 - \sum_{i \in [d]} p_i^2}^{1/2}$.

Our work shares some similarities with the quantum gradient estimation algorithm of Jordan~\cite{Jor05j,GAW19c}, which also uses an extension of the Bernstein-Vazirani algorithm to linear functions over the reals. However, unlike gradient estimation, the mean estimation problem requires combining this technique with further algorithmic steps.

%% file: Notations.tex
Throughout the paper we use the $\wbo{x}$ and $\wom{x}$ notations to hide factors that are polylogarithmic in the argument~$x$. We let $\ket{\vec{0}}$ denote a multiple qubits state $\ket{0\dots 0}$. We use the notation~$\Hila$ when referring to an auxiliary Hilbert space of sufficiently large dimension. We consider the family of $\ell_p$-norms, defined as follows.

\begin{definition}[\sc $\ell_p$-norm]
  Given $p \in [1,+\infty)$, the \emph{$\ell_p$-norm} $\norm{x}_p$ of a $d$-dimensional vector $x$ is defined as $\norm{x}_p = \pt{\sum_{i \in [d]} \abs{x_i}^p}^{1/p}$ if $p < \infty$, and $\norm{x}_{\infty} = \max_{i \in [d]} \abs{x_i}$. We also let $\norm{x} = \norm{x}_2$ denote the $\ell_2$-norm, and for a matrix $M$ we set $\norm{M}$ to be the induced $\ell_2$-norm (or \emph{spectral norm}).
\end{definition}

Given $x \in \R^d$ and $0 \leq a < b$, we define the following truncation with respect to the $\ell_2$-norm.
  \[\clamp{x}{a}{b} =
    \left\{
       \begin{array}{l}
         x \quad \text{if $a < \norm{x}_2 \leq b$,} \\[2mm]
         0 \quad \text{otherwise.}
       \end{array}
    \right.\]

We recall the definition of a multivariate random variable. We only consider finite probability spaces for finite encoding reasons. Throughout the paper $d \in \N$ will denote the dimension of the random variable whose mean is to be estimated.

\begin{definition}[\sc Random variable]
  \label{Def:RandVar}
  A (finite) \emph{random variable} is a function $X : \Omega \ra E$ for some probability space $(\Omega, 2^{\Omega}, \P)$, where $\Omega$ is a finite sample set, $\P : \Omega \ra [0,1]$ is a probability mass function and $E \subset \R^d$ is the finite support of $X$. The \emph{covariance matrix} $\Sigma \in \R^{d \times d}$ of $X$ is defined as $\Sigma = \ex{X X^{\top}} - \ex{X}\ex{X}^{\top}$.
\end{definition}

We say that $X$ is \emph{univariate} if the dimension is $d = 1$, and multivariate otherwise. For any norm~$\norm{.}$ over $\R^d$, we let~$\norm{X}$ denote the univariate random variable $\omega \mapsto \norm{X(\omega)}$. Finally, we recall the definition of a quantile value (using the complementary cumulative distribution function).

\begin{definition}[\sc Quantile]
  \label{Def:quantile}
  Given a discrete real-valued random variable $X$ and a real $p \in [0,1]$, the \emph{quantile} of order $p$ is the number
    $Q(p) = \sup\set{x \in \R : \pr{X \geq x} \geq p}$.
\end{definition}


%% file: Models.tex
The input to the multivariate mean estimation problem is represented by a random variable~$X$ taking values in $\R^d$.  In this section, we describe two possible access models for quantum estimators. Before that, we first recall the classical model, which we refer to as a \emph{random experiment}.


\begin{definition}[\sc Random experiment]
  \label{Def:rExp}
  Given a random variable~$X$ on a probability space $(\Omega, 2^{\Omega}, \P)$, we define a \emph{random experiment} as the process of drawing a sample $\omega \in \Omega$ according to~$\P$ and observing the value of $X(\omega) \in \R^d$.
\end{definition}

In the quantum setting, we make a distinction between accessing the probability mass function~$\P$ and evaluating the function $X : \Omega \ra E$. The first operation is provided by means of a \emph{quantum experiment}, defined in the following way.

\begin{definition}[\sc Quantum experiment]
  \label{Def:qExp}
  Consider a random variable $X$ on a probability space $(\Omega, 2^{\Omega}, \P)$. Let $\Hil_{\Omega}$ be a Hilbert space with basis states $\set{\ket{\omega}}_{\omega \in \Omega}$ and fix a unitary $U_{\P}$ acting on $\Hil_{\Omega}$ such that
    \[U_{\P} : \ket{0} \mapsto \sum_{\omega \in \Omega} \sqrt{\P(\omega)} \ket{\omega}\]
  assuming $0 \in \Omega$. We define a \emph{quantum experiment} as the process of applying the unitary~$U_{\P}$ or its inverse~$U_{\P}^{-1}$ on any state in~$\Hil_{\Omega}$.
\end{definition}

We note that $\P$ is sometimes assumed to be the uniform distribution over some large set~$\Omega = [N]$ (e.g.~\cite{Gro98c,NW99c,Hei02j,BHH11j,CFMdW10c,BDGT11p,LW19j}). In this case, the access to the unitary $U_{\P}$ need not be provided as part of the input.

We now describe two different quantum oracles for evaluating $X$. The first oracle provides a direct access to the value of $X(\omega)$. This model is the most commonly used in previous work on quantum mean estimation (e.g. \cite{Gro98c,Ter99d,NW99c,Hei02j,BDGT11p,Mon15j,HM19c}).

\begin{definition}[\sc Binary oracle]
  \label{Def:binOracle}
  Consider a finite random variable $X : \Omega \to E$ on a probability space $(\Omega, 2^{\Omega}, \P)$. Let $\Hil_{\Omega}$ and $\Hil_E$ be two Hilbert spaces with basis states $\set{\ket{\omega}}_{\omega \in \Omega}$ and $\set{\ket{x}}_{x \in E}$ respectively. We say that a unitary $\bora_X$ acting on $\Hil_{\Omega} \otimes \Hil_E$ is a \emph{binary oracle} for $X$ if
    \[\bora_X : \ket{\omega}\ket{\vec{0}} \mapsto \ket{\omega}\ket{X(\omega)}\]
  for all $\omega \in \Omega$, assuming $\vec{0} \in E$.
\end{definition}

Observe that one random experiment can be simulated by preparing the state $\sum_{\omega \in \Omega} \sqrt{\P(\omega)} \ket{\omega} \allowbreak \ket{X(\omega)}$ and measuring its last register in the $\set{\ket{x}}_{x \in E}$ basis. This requires using one quantum experiment and one call to the binary oracle.

Our second type of oracle provides individual access to the coordinates of $X(\omega)$, encoded into the phases of a query operator. This model appears naturally in the context of variational eigensolvers, QAOA, and training variational auto-encoders \cite{GAW19c}, albeit not in relation to the quantum mean estimation problem. This input model can be efficiently simulated using a binary oracle, but the converse is generally not true. In fact, even obtaining one classical sample from $X$ may not be easy to do using a phase oracle.

\begin{definition}[\sc Phase oracle]
  \label{Def:phaseOracle}
  Consider a finite random variable $X : \Omega \to E$ on a probability space $(\Omega, 2^{\Omega}, \P)$. Let $\Hil_{\Omega}$ be a Hilbert space with basis states $\set{\ket{\omega}}_{\omega \in \Omega}$. We say that a unitary~$\pora_X$ acting on $\Hil_{\Omega} \otimes \C^d$ is a \emph{phase oracle} for $X$ if
    \[\pora_X : \ket{\omega}\ket{j} \mapsto e^{i X(\omega)_j} \ket{\omega}\ket{j}\]
  for all $\omega \in \Omega$ and $j \in [d]$.
\end{definition}

%% file: Tools.tex
We first recall the optimal classical error bound for estimating the mean of a multivariate random variable with respect to the Euclidean norm.

\begin{proposition}[\sc Classical sub-Gaussian estimators, \cite{LM19j}]
  \label{Prop:ClassSubGaussian}
  Let $X$ be a $d$-dimensional random variable with mean $\mu$ and covariance matrix $\Sigma$. Given $\delta \in (0,1)$ and $n \geq \log(1/\delta)$, the \emph{sub-Gaussian estimators} outputs a mean estimate $\mut$ such that
    \[\norm{\mut - \mu}_2 \leq \sqrt{\frac{\Tr(\Sigma)}{n}} + \sqrt{\frac{\norm{\Sigma} \log(1/\delta)}{n}}\]
  with probability at least $1-\delta$, by using $\bo{n}$ random experiments.
\end{proposition}

We now present four quantum subroutines used in our work. We first need an algorithm introduced in~\cite{DH96p,NW99c} and generalized in~\cite{Ham21c} for estimating the quantiles (Definition~\ref{Def:quantile}) of a univariate random variable quadratically faster than it is possible classically.

\begin{proposition}[\sc Quantile estimator, \cite{Ham21c}]
  \label{Prop:Quantile}
  Let $X$ be a univariate random variable. Given two reals $p, \delta \in (0,1)$, the \emph{quantile estimation} algorithm $\quant(X,p,\delta)$ returns an approximate quantile $\td{Q}$ that satisfies
    \[Q(p) \leq \td{Q} \leq Q(cp)\]
  with probability at least $1-\delta$ for some universal constant $c \in (0,1)$. The algorithm uses $\bo[\Big]{\frac{\log(1/\delta)}{\sqrt{p}}}$ quantum experiments and binary oracle queries to $X$.
\end{proposition}

Next, we will use a variant of amplitude amplification~\cite{BHMT02j} that provides a precise linear amplification of the amplitude.

\begin{proposition}[{\sc Linear amplitude amplification}, Theorem 6.10 in \cite{Low17d} or Lemma 11 in \cite{GL20c}]
  \label{Prop:AA-SVT}
  Let~$V$ be a unitary operator and let $\Pi$ be a projection operator acting on the same Hilbert space. Given two reals $t \geq 1$ and $\eps \in (0,1)$ there is a unitary operator $V_{t,\eps}$ that can be implemented with $\bo{t \log(1/\eps)}$ applications of $V$, $V^\dagger$ and $I - 2\Pi$, and such that
    \[\abs[\big]{\norm{\Pi V_{t,\eps} \qub} - t \norm{\Pi V \qub}} \leq \eps  \quad \text{if} \quad t \norm{\Pi V \qub} \leq 1/2.\]
\end{proposition}


Finally, the next two results provide efficient algorithms for converting between phase and amplitude encodings.

\begin{lemma}[{\sc Amplitude-to-Phase conversion}, Corollary 4.1 in \cite{GAW19c}]
  \label{Lem:PO}
  Let~$V$ be a unitary operator acting on some Hilbert space $\Hil_U \otimes \Hil$ such that
    \[V : \ket{u}\qub \mapsto \ket{u}\pt{\sqrt{1-p_u} \ket{\psi_u^0}\ket{0} + \sqrt{p_u} \ket{\psi_u^1}\ket{1}}\]
  where $\set{\ket{u}}_{u \in U}$ is the standard basis of $\Hil_U$, $p_u \in (0,1)$ and $\ket{\psi_u^0}$, $\ket{\psi_u^1}$ are some arbitrary unit states. Then, given two reals $t \geq 0$ and $\eps \in (0,1)$, there is a unitary operator $\pora_{t,\eps}$ acting on $\Hil_U \otimes \Hil \otimes \Hila$ that can be implemented with $\bo{t + \log(1/\eps)}$ applications of $V$ and~$V^\dagger$, such that
    \[\pora_{t,\eps} : \ket{u}\qub \mapsto \ket{u} \ket{\varphi_u} \quad \text{where} \quad \norm{\ket{\varphi_u} - e^{it p_u}\qub} \leq \eps,\]
  for all $u \in U$.
\end{lemma}

\begin{lemma}[{\sc Phase-to-Amplitude conversion}, Lemma~16 in \cite{GAW17p}]
  \label{Lem:PrO}
  Let~$\pora$ be a unitary operator acting on some Hilbert space $\Hil_U$ such that
    \[\pora : \ket{u} \mapsto e^{ip_u}\ket{u}\]
  where $\set{\ket{u}}_{u \in U}$ is the standard basis of $\Hil_U$ and $p_u \in [\delta,1-\delta]$ for some $\delta \in (0,1/2)$. Then, given a real $\eps \in (0,1)$, there is a unitary operator $V_{\eps,\delta}$ acting on $\Hil_U \otimes \Hila \otimes \C^2$ that can be implemented with $\bo{\log(1/\eps)/\delta}$ applications of $\pora$ and $\pora^{\dagger}$, such that
   \[V_{\eps,\delta} : \ket{u}\ket{\vec{0}}\ket{0} \mapsto \ket{u} \pt[\big]{\sqrt{p'_u}\ket{\vec{0}}\ket{0} + \sqrt{1-p'_u}\ket{\psi}\ket{1}} \quad \text{where} \quad \abs[\big]{\sqrt{p'_u} - \sqrt{p_u}} \leq \eps,\]
  for all $u \in U$ and some state $\ket{\psi}$.
\end{lemma}

%% file: BoundedBinary.tex
In this section, we generalize the univariate bounded estimator~\cite{Ter99d,WCNA09j,Mon15j} derived from Amplitude Estimation~\cite{BHMT02j} to the multivariate setting $X \in \R^d$. Our main ingredient is the construction of an approximate phase oracle for the directional mean $\inp{u}{\ex{X}}$, where the vectors $u \in \R^d$ are selected from the grid of points,
  \[G = \set*{\frac{j}{\nn} - \frac{1}{2} + \frac{1}{2\nn} : j \in \set{0,\dots,\nn-1}}^d \subset (-1/2,1/2)^d\]
with $\nn$ being defined in step~\ref{Step:mdefine} of Algorithm~\ref{Alg:MultiBoundedEst}.

We let $u \sim G$ denote a vector obtained according to the uniform distribution over $G$. We also define $\Hil_G$ to be the Hilbert space whose standard basis is indexed by the elements of $G$. Our algorithm requires encoding the inner product $\inp{u}{X}$ into an amplitude. However, this quantity can be as large as $\sqrt{d}$ assuming $\norm{X}_2 \leq 1$. The next crucial result shows that it is in fact much smaller than~$\sqrt{d}$ for most values of $u$.


\begin{lemma}
  \label{Lem:Truncate}
  Let $\alpha > 0$. For any vector $x \in \R^d$ and any random variable $X$ over $\R^d$ we have,
    \[\pr_{u \sim G}[\big]{\alpha  \abs{\inp{u}{x}} \geq \norm{x}_2} \leq 2e^{-2/\alpha^2} \quad \mathrm{and} \quad \pr_{u \sim G}[\big]{\alpha  \ex{\abs{\inp{u}{X}}} \geq \ex{\norm{X}_2}} \leq \alpha/2.\]
\end{lemma}

\begin{proof}
  We use that the coordinates of a uniformly random vector $u \in G$ are independent centered random variables bounded in $(-1/2,1/2)$. The first result is obtained using Hoeffding's inequality, $\pr_u{\alpha \abs{\inp{u}{x}} \geq \norm{x}_2} \leq 2\exp\pt[\Big]{\frac{-2 \norm{x}_2^2}{\sum_{j=1}^d \abs{\alpha x_j}^2}} = 2e^{-2/\alpha^2}$, since $\ex_u{\inp{u}{x}} = 0$ for all $x \in \R^d$. For the second result, we have by Markov's inequality that $\pr_{u \sim G}[\big]{\alpha \ex{\abs{\inp{u}{X}}} \geq \ex{\norm{X}_2}} \leq \frac{\ex_u{\alpha\ex{\abs{\inp{u}{X}}}}}{\ex{\norm{X}_2}} = \frac{\alpha \ex{\ex_u{\abs{\inp{u}{X}}}}}{\ex{\norm{X}_2}}$. By Cauchy-Schwarz inequality, $\ex_u{\abs{\inp{u}{x}}} \leq \sqrt{\ex_u{\inp{u}{x}^2}} = \sqrt{\sum_{j=1}^d \ex_{u_j}{(u_jx_j)^2}} \leq \norm{x}_2/2$ for all $x \in \R^d$. Thus, $\pr_{u \sim G}*{\alpha  \ex{\abs{\inp{u}{X}}} \geq \ex{\norm{X}_2}} \leq \alpha/2$.
\end{proof}

Using the above lemma, one can encode (for most values of $u$) the truncated directional mean $\ex{\clamp{\alpha \inp{u}{X}}{0}{1}}$ into an amplitude and apply oracle conversion techniques to approximate the phase oracle $\ket{u} \mapsto e^{i \ex{\clamp{\alpha \inp{u}{X}}{0}{1}}} \ket{u}$ with accuracy $\eps$ at cost $\bo{\log(1/\eps)}$. The cost of applying $m$ times this oracle is then $\bo{m \log(1/\eps)}$. We describe a more subtle algorithm where the latter complexity becomes $\wbo{m \sqrt{L_2} \log^2(1/\eps)}$ given an upper-bound $L_2 \geq \ex{\norm{X}_2}$. The dependence on $\ex{\norm{X}_2}$ generalizes the dependence on $\ex{\abs{X}}$ provided by the univariate bounded estimator~\cite{Ter99d,BHMT02j,WCNA09j,Mon15j}.

\begin{proposition}[\sc Directional mean oracle]
  \label{Prop:PhaseExp}
  Let $X$ be a $d$-dimensional bounded random variable such that $\norm{X}_2 \leq 1$. Given four reals $L_2 \in (0,1]$, $\nn \geq 1/L_2$, $\alpha, \eps \in (0,1)$ such that $\ex{\norm{X}_2} \leq L_2$, there exists a unitary operator $\td{\pora}_{X,L_2,\nn,\alpha,\eps} : \ket{u}\qub \mapsto \ket{u} \ket{\varphi_u}$ acting on $\Hil_G \otimes \Hila$ that can be implemented using
    $\wbo[\big]{\nn\sqrt{L_2} \log^2(1/\eps)}$ 
  quantum experiments and binary oracle queries to~$X$, and such that
    \[\norm*{\ket{\varphi_u} - e^{i \nn\ex*{\clamp{\alpha \inp{u}{X}}{0}{1}}}\qub} \leq \eps\]
  for a fraction at least $1-\alpha/2$ of all $u \in G$.
\end{proposition}

\begin{proof}
  Fix $u$ and consider the random variable $X_+$ defined over the same probability space as~$X$ such that $X_+(\omega) = X(\omega)$ when $\alpha \inp{u}{X(\omega)} > 0$ and $X_+(\omega) = 0$ otherwise. Similarly, define $X_-$ such that $X_-(\omega) = X(\omega)$ if $\alpha \inp{u}{X(\omega)} < 0$ and $X_-(\omega) = 0$ otherwise. Since $\ex*{\clamp{\alpha \inp{u}{X}}{0}{1}} = \ex*{\clamp{\alpha \inp{u}{X_+}}{0}{1}} + \ex*{\clamp{\alpha \inp{u}{X_-}}{0}{1}}$ it is sufficient to explain how to construct a unitary $\pora_+ : \ket{u}\qub \mapsto \ket{u}\ket{\varphi_{+,u}}$ such that $\norm[\big]{\ket{\varphi_{+,u}} - e^{i \nn \ex*{\clamp{\alpha \inp{u}{X_+}}{0}{1}}}\qub} \leq \eps/2$ when $\alpha \ex{\abs{\inp{u}{X}}} \leq L_2$. One can construct $\pora_-$ that encodes $\ex*{\clamp{\alpha \inp{u}{X_-}}{0}{1}}$ using a similar approach. The proposition then follows by taking the product $\td{\pora}_{X,L_2,\nn,\alpha,\eps} = \pora_+\pora_-$ and noting that $\pr_u{\alpha \ex{\abs{\inp{u}{X}}} \geq L_2} \leq \alpha/2$ by the second part of Lemma~\ref{Lem:Truncate} since $\ex{\norm{X}_2} \leq L_2$.

  There are three steps in the construction of $\pora_+$. First, we construct a unitary $V_+$ acting on $\Hil_G \otimes \Hil_{\Omega} \otimes \Hil_{E} \otimes \C^2$ as follows,
    \begin{align*}
      V_+ : \ket{u}\qub\ket{0}
            & \mapsto \sum_{\omega \in \Omega} \sqrt{\P(\omega)} \ket{u}\ket{\omega,X(\omega)}\ket{0} \\
            &\mapsto \sum_{\omega \in \Omega} \sqrt{\P(\omega)} \ket{u}\ket{\omega,X(\omega)}\pt[\Big]{\sqrt{1 - \clamp{\alpha \inp{u}{X_+(\omega)}}{0}{1}}\ket{0} + \sqrt{\clamp{\alpha \inp{u}{X_+(\omega)}}{0}{1}}\ket{1}} \\
            & = \sqrt{1 - \ex{\clamp{\alpha \inp{u}{X_+}}{0}{1}}} \ket{u}\ket{\psi_u^0}\ket{0} + \sqrt{\ex{\clamp{\alpha \inp{u}{X_+}}{0}{1}}} \ket{u}\ket{\psi_u^1}\ket{1}
    \end{align*}
  where the first step uses one quantum experiment and one binary oracle query to $X$, the second step performs a sequence of controlled rotations, and in the last line $\ket{\psi_u^0}, \ket{\psi_u^1}$ are some irrelevant unit states. Secondly, if $L_2 < 1/4$, we apply the Linear Amplitude Amplification algorithm of Proposition~\ref{Prop:AA-SVT} on $V_+$ with $t = 1/(2\sqrt{L_2})$ and accuracy $\eps/(32\nn L_2)$, which gives a unitary $W_+ : \ket{u}\qub\ket{0} \mapsto \sqrt{1 - p_u} \ket{u}\ket{\psi_u^0}\ket{0} + \sqrt{p_u} \ket{u}\ket{\psi_u^1}\ket{1}$ where, for each $u \in G$,
    \[\abs*{\sqrt{p_u} - \sqrt{\frac{\ex{\clamp{\alpha \inp{u}{X_+}}{0}{1}}}{4L_2}}} \leq \frac{\eps}{32\nn L_2}, \tag*{if $\ex{\clamp{\alpha \inp{u}{X_+}}{0}{1}} \leq L_2$,}\]
  using $\bo{\log(\nn L_2/\eps)/\sqrt{L_2}}$ applications of $V_+$ and $V_+^{\dagger}$. If $L_2 < 1/4$ we directly take $W_+ = V_+$. Thirdly, we define $L'_2 = \min\pt{L_2,1/4}$ and apply the phase conversion algorithm of Lemma~\ref{Lem:PO} on~$W_+$ with $t = 4\nn L'_2$ and accuracy $\eps/4$. We obtain a unitary $\pora_+ : \ket{u}\qub \mapsto \ket{u}\ket{\varphi_{+,u}}$ such that, by the triangle inequality, the state $\ket{\varphi_{+,u}}$ satisfies
    $\norm[\big]{\ket{\varphi_{+,u}} - e^{i \nn \ex{\clamp{\alpha \inp{u}{X_+}}{0}{1}}}\qub}
       \leq \norm[\big]{\ket{\varphi_{+,u}} - e^{i 4mL'_2 p_u}\qub} + \abs[\big]{e^{i 4mL'_2 p_u} - e^{i \nn \ex{\clamp{\alpha \inp{u}{X_+}}{0}{1}}}}
       \leq \eps/4 + \abs[\big]{4mL'_2 p_u - m\ex{\clamp{\alpha \inp{u}{X_+}}{0}{1}}}
       = \eps/4 + 4mL'_2 \abs[\Big]{\sqrt{p_u} - \sqrt{\frac{\ex{\clamp{\alpha \inp{u}{X_+}}{0}{1}}}{4L'_2}}} \cdot \abs[\Big]{\sqrt{p_u} + \sqrt{\frac{\ex{\clamp{\alpha \inp{u}{X_+}}{0}{1}}}{4L'_2}}}$.
   Thus, for each $u \in G$,
    \[\norm*{\ket{\varphi_{+,u}} - e^{i \nn \ex*{\clamp{\alpha \inp{u}{X_+}}{0}{1}}}\qub} \leq \eps/2, \tag*{if $\ex{\clamp{\alpha \inp{u}{X_+}}{0}{1}} \leq L_2$,}\]
  and $\pora_+$ uses $\bo{\nn L_2 + \log(1/\eps)}$ applications of $W_+$ and $W_+^{\dagger}$. Overall, we used $\wbo{\nn\sqrt{L_2}\log^2(1/\eps)}$ quantum experiments and binary oracle queries to $X$ to implement one application of $\pora_+$. Moreover, the condition $\ex{\clamp{\alpha \inp{u}{X_+}}{0}{1}} \leq L_2$ is implied by $\alpha \ex{\abs{\inp{u}{X}}} \leq L_2$.
\end{proof}

We finally describe the algorithm that estimates the mean of any bounded random variable (Algorithm~\ref{Alg:MultiBoundedEst}). Our approach relies on applying the quantum Fourier transform over $G$ to the above directional mean oracle in a similar manner as in previous work (e.g.~\cite{BV97j,Jor05j,GAW19c}). The results of Lemma~\ref{Lem:Truncate} play again a central role in the analysis.

\algobox{Alg:MultiBoundedEst}{Bounded multivariate estimator, $\boundedest_d\pt{X,L_2,n,\delta}$.}{
\begin{enumerate}[leftmargin=*]
  \item If $n \leq \frac{\log(d/\delta)}{\sqrt{L_2}}$ then output $\mut = 0$.\label{Step:trivial}
  \item Set $\alpha = \frac{1}{\sqrt{\log(400\pi n \sqrt{d})}}$ and $\nn = 2^{\ceil[\big]{\log\pt[\big]{\frac{8\pi}{\alpha} \cdot \frac{n}{\sqrt{L_2} \log(d/\delta)}}}}$.
  \item For $k = 1,\dots,\ceil{18\log(d/\delta)}$:
  \begin{enumerate}
    \item Compute the uniform superposition $\ket{G} := \frac{1}{\nn^{d/2}} \sum_{u \in G} \ket{u}$ over $G$.
    \item Compute the state $\ket{\psi} := \td{\pora}_{X,L_2,\nn,\alpha,\eps} \ket{G}\qub \in \Hil_G \otimes \Hila$, where $\td{\pora}_{X,L_2,\nn,\alpha,\eps}$ is the directional mean oracle constructed in Proposition~\ref{Prop:PhaseExp} with $\eps = 1/25$.\label{Step:direc}
    \item Compute the state $\ket{\phi} := \pt{\mathrm{QFT}^{-1}_G \otimes \id_{\mathrm{aux}}}\ket{\psi}$ where the unitary  $\mathrm{QFT}_G : \ket{u} \mapsto \frac{1}{\nn^{d/2}} \sum_{v \in G} e^{2i\pi \nn \inp{u}{v}} \ket{v}$ is the quantum Fourier transform over~$G$.\label{Step:qft}
    \item Measure the $\Hil_G$ register of $\ket{\phi}$ in the computational basis and let $\td{v}\super{k} \in G$ denote the obtained result. Set $\mut\super{k} = \frac{2\pi}{\alpha} \td{v}\super{k}$.\label{Step:measure}
  \end{enumerate}
  \item Output the coordinate-wise median $\mut = \median\pt{\mut\super{1},\dots,\mut\super{{\ceil{18\log(d/\delta)}}}}$.
\end{enumerate}
}

\begin{theorem}[\sc Bounded multivariate estimator]
  \label{Thm:MultiBoundedEst}
  Let $X$ be a $d$-dimensional bounded random variable such that $\norm{X}_2 \leq 1$. Given three reals $L_2 \in (0,1]$, $\delta \in (0,1)$ and $n \geq 1$ such that $\ex{\norm{X}_2} \leq L_2$, the \emph{bounded multivariate estimator} $\boundedest_d\pt{X,L_2,n,\delta}$ (Algorithm~\ref{Alg:MultiBoundedEst}) outputs a mean estimate $\mut$ of $\mu = \ex{X}$ such that
    \[\norm{\mut - \mu}_{\infty} \leq \frac{\sqrt{L_2} \log(d/\delta)}{n}\]
  with probability at least $1-\delta$. It uses $\wbo{n}$ quantum experiments and binary oracle queries to $X$.
\end{theorem}

\begin{proof}
  If $n \leq \frac{\log(d/\delta)}{\sqrt{L_2}}$ then by choosing $\mut = 0$ at step~\ref{Step:trivial} we directly have $\norm{\mut - \mu}_{\infty} \leq \ex{\norm{X}_2} \leq \frac{\sqrt{L_2} \log(d/\delta)}{n}$. Thus, we can suppose from now on that $n > \frac{\log(d/\delta)}{\sqrt{L_2}}$ (in particular, $m \geq 1/L_2$).

  Let $\ket{\psi'}, \ket{\psi''} \in \Hil_G \otimes \Hila$ be the two unit states defined as follows,
    \[\ket{\psi'} = \frac{1}{\nn^{d/2}} \sum\nolimits_{u \in G} e^{i \nn\ex*{\clamp{\alpha\inp{u}{X}}{0}{1}}} \ket{u} \qub
      \quad \text{and} \quad
      \ket{\psi''} = \frac{1}{\nn^{d/2}} \sum\nolimits_{u \in G} e^{i \nn \alpha \inp{u}{\ex{X}}} \ket{u} \qub.\]
  We first show that the state~$\ket{\psi}$ at step~\ref{Step:direc} satisfies $\norm{\ket{\psi} - \ket{\psi''}} \leq 1/12$. On one hand, by Proposition~\ref{Prop:PhaseExp}, we have
    $\norm{\ket{\psi} - \ket{\psi'}}^2
      = \frac{1}{\nn^d} \sum_u \norm{\td{\pora}_{X,L_2,\nn,\alpha,\eps}\pt{\ket{u}\qub} - e^{i \nn\ex*{\clamp{\alpha \inp{u}{X}}{0}{1}}}\ket{u}\qub}^2
      \leq \eps^2 + \alpha$.
  On the other hand, by using the inequality $\sin^2(x) \leq \abs{x}$, we have
    $\norm{\ket{\psi'} - \ket{\psi''}}^2
      = \frac{4}{\nn^d} \sum_{u \in G} \sin^2\pt*{\frac{\nn}{2}\pt[\big]{\ex*{\clamp{\alpha\inp{u}{X}}{0}{1}} - \alpha\inp{u}{\ex{X}}}}
      \leq \frac{2\nn}{\nn^d} \sum_{u \in G} \abs[\big]{\ex*{\clamp{\alpha\inp{u}{X}}{0}{1}} - \alpha\inp{u}{\ex{X}}}
      \leq \frac{2\nn}{\nn^d} \allowbreak \ex*{\sum_u \abs[\big]{\clamp{\alpha\inp{u}{X}}{0}{1} - \alpha\inp{u}{X}}}
      \leq 2\nn \alpha \sqrt{d} \pr_u{\alpha\abs{\inp{u}{X}} > 1}$
  where the last step uses $\abs*{\inp{u}{X}} \leq \sqrt{d}$. By the first part of Lemma~\ref{Lem:Truncate}, we have $\pr_u[\big]{\alpha\abs{\inp{u}{X}} \geq 1} \leq 2e^{-2/\alpha^2}$ since $\norm{X}_2 \leq 1$. Thus, $\norm{\ket{\psi'} - \ket{\psi''}}^2 \leq 4 \nn \alpha \sqrt{d} e^{-2/\alpha^2}$. By the triangle inequality, $\norm{\ket{\psi} - \ket{\psi''}} \leq \sqrt{\eps^2 + \alpha} + 2\sqrt{\nn \alpha} d^{1/4} e^{-1/\alpha^2} \leq 1/12$.

  We now analyse steps~\ref{Step:qft} and~\ref{Step:measure} of the algorithm assuming $\ket{\psi}$ is replaced with $\ket{\psi''}$. Let $\td{v} \in G$ denote the vector obtained by measuring the first register of $\pt{\mathrm{QFT}^{-1}_G \otimes \id_{\mathrm{aux}}} \ket{\psi''}$ in the computational basis. Since the phases satisfy $\norm{\alpha \ex{X}}_{\infty} \leq 2\pi/3$, we can apply the analysis of phase estimation (e.g. \cite[Lemma 5.1]{GAW19c}) to conclude that $\abs[\big]{\td{v}_j - \frac{\alpha}{2\pi} \ex{X}_j} \leq 4/\nn$ with probability at least $5/6$ for each $j \in [d]$. By replacing $\ket{\psi}$ with $\ket{\psi''}$, we achieve the same result with probability at least $5/6 - 2\norm{\ket{\psi} - \ket{\psi''}} \geq 2/3$. Finally, by the Chernoff bound, $\norm[\big]{\mut - \ex{X}}_{\infty} \leq 8\pi/(\alpha\nn) \leq \sqrt{L_2} \log(d/\delta)/n$ with probability at least $1-\delta$.

  The total number of quantum experiments and binary oracle queries to $X$ is $\bo{\log(d/\delta) \cdot \nn\sqrt{L_2}\log^2(1/\eps)} = \bo{n/\alpha} = \wbo{n}$.
\end{proof}

%% file: MultiBinary.tex
We describe in Algorithm~\ref{Alg:MultiEstim} our main quantum algorithm for estimating the mean of a $d$-dimensional random variable in the binary oracle model. Our approach relies on applying the bounded multivariate estimator, developed in the previous section, to a sequence of carefully chosen truncated random variables.

\algobox{Alg:MultiEstim}{Near-optimal multivariate estimator, $\qestim_d(X,n,\delta)$.}{
\begin{enumerate}[leftmargin=*]
  \item Set $k = \ceil{2 \log\pt[\big]{\frac{2\sqrt{2} n}{\log(d/\delta)}}}$ and $n' = n \cdot \frac{(k+1)4\log(5kd/\delta)}{\sqrt{c} \log(d/\delta)}$ where $c$ is the constant mentioned in Theorem~\ref{Prop:Quantile}.
  \item Run any classical sub-Gaussian estimator (Proposition~\ref{Prop:ClassSubGaussian}) on $X$ using $\bo{\log(1/\delta)}$ samples to compute a mean estimate $\eta \in \R^d$ such that $\pr[\big]{\norm{\eta - \mu}_2 > \sqrt{\Tr(\Sigma)}} \leq \delta/2$.\label{Step:mdefine}
  \item Define the random variable $Y = X - \eta$.
  \item For $j = 0,\dots,k$:\label{Step:For}
  \begin{enumerate}
    \item Compute an estimate $a_{j}$ of the quantile of order $2^{-j}$ of $\norm{Y}_2$ by using the \hyperref[Prop:Quantile]{quantile estimator} $\quant(\norm{Y}_2,2^{-j},\delta/(5k))$.\label{Step:quantile}
    \item Define the bounded random variable $Y_{j} = \frac{1}{{a_{j}}} \clamp{Y}{a_{j - 1}}{a_{j}}$ (where $a_{-1} = 0$). If $a_{j - 1} = a_{j}$ then set $\mut_{j} = 0$, else compute an estimate $\mut_{j}$ of $\ex{Y_{j}}$ by using the \hyperref[Thm:MultiBoundedEst]{bounded multivariate estimator} $\boundedest_d(Y_{j},2^{-(j-1)},n',\delta/(5k))$.
  \end{enumerate}
  \item Output $\mut = \eta + \sum_{j = 0}^{k} a_{j} \mut_{j}$.
\end{enumerate}
}

\begin{theorem}[\sc Near-optimal multivariate estimator]
  \label{Thm:MultiEstim}
  Let $X$ be a $d$-dimensional random variable with mean $\mu$ and covariance matrix $\Sigma$. Given two reals $\delta \in (0,1)$ and $n \geq \log(d/\delta)$, the quantum multivariate estimator $\qestim_d\pt{X,n,\delta}$ (Algorithm~\ref{Alg:MultiEstim}) outputs a mean estimate $\mut$ such that
    \[\norm{\mut - \mu}_{\infty} \leq \frac{\sqrt{\Tr(\Sigma)} \log(d/\delta)}{n}\]
  with probability at least $1-\delta$. It uses $\wbo{n}$ quantum experiments and binary oracle queries to $X$.
\end{theorem}

\begin{proof}
  The main part of the proof is to show that the mean estimate $\mut_Y = \sum_{j = 0}^{k} a_{j} \mut_{j}$ of $\mu_Y = \ex{Y}$ satisfies
    \begin{equation}
      \label{Eq:MultiBound}
      \norm{\mut_Y-\mu_Y}_{\infty} \leq \frac{\sqrt{\ex{\norm{Y}^2_2}} \log(d/\delta)}{\sqrt{2}n}
    \end{equation}
  with probability at least $1-\delta/2$. The theorem follows since $\norm{\mut - \mu}_{\infty} = \norm{\mut_Y-\mu_Y}_{\infty}$ and $\ex{\norm{Y}^2_2} = \ex{\norm{X-\mu}_2^2} + \norm{\mu - \eta}^2_2 = \Tr(\Sigma) +  \norm{\mu - \eta}^2_2 \leq 2 \Tr(\Sigma)$, where the last inequality holds with probability at least $1-\delta/2$. The algorithms uses $\wbo{k 2^{k/2} \log(k/\delta) + k n'} = \wbo{n}$ quantum experiments and binary oracle queries to $X$.

  We now turn to the proof of Equation~(\ref{Eq:MultiBound}). We make the assumption that all the subroutines used in step~\ref{Step:For} are successful, which is the case with probability at least $(1-\delta/(5k))^{2k+2} \geq 1 - \delta/2$. The sequence $(a_{j})_{j}$ of quantile estimates computed at step~\ref{Step:quantile} satisfies $Q(2^{-j}) \leq a_{j} \leq Q(c2^{-j})$ for all $j \in \set{0,\dots,k}$, where $c$ is the constant mentioned in Theorem~\ref{Prop:Quantile}. On one hand, by Markov's inequality, $\pr{\norm{Y}_2 \geq a_{j}} = \pr{\norm{Y}^2_2 \geq a_{j}^2} \leq \ex{\norm{Y}^2_2}/a_{j}^2$. On the other hand, by definition of the quantile function, $\pr{\norm{Y}_2 \geq a_{j}} \geq \pr{\norm{Y}_2 \geq Q(c2^{-j})} \geq c2^{-j}$. Thus,
    \begin{equation}
        \label{Eq:quantBound}
        a_{j} \leq c^{-1/2} 2^{j/2} \sqrt{\ex{\norm{Y}^2_2}}.
    \end{equation}
  Since $\pr{\norm{Y}_2 > a_{j-1}} \leq \pr{\norm{Y}_2 > Q(2^{-(j-1)})} < 2^{-(j-1)}$, we also have that
    \begin{equation}
        \label{Eq:MomBound}
        \ex{\norm{Y_{j}}_2} < 2^{-(j-1)}.
    \end{equation}
  Hence, $2^{-(j-1)}$ is a valid upper bound on the expectation of $\norm{Y_{j}}_2$. Consequently, by Theorem~\ref{Thm:MultiBoundedEst}, each estimate $\mut_{j}$ satisfies $\norm{\mut_{j}-\ex{Y_{j}}}_{\infty} \leq \frac{2^{-(j-1)/2}\log(5kd/\delta)}{n'}$. Moreover, the truncated random variable $\clamp{Y}{a_k}{+\infty}$ satisfies,
    \begin{equation}
      \label{Eq:truncLargest}
      \norm{\clamp{Y}{a_k}{+\infty}}_{\infty} \leq \norm{\clamp{Y}{a_k}{+\infty}}_2 \leq \sqrt{\ex{\norm{Y}_2^2} \pr{\norm{Y}_2 > a_k}} \leq \frac{\sqrt{\ex{\norm{Y}^2_2}}}{2^{k/2}}
    \end{equation}
  where the first step is by monotonicity of the norm and the second is by Cauchy-Schwartz inequality.
  Overall, the error is $\norm{\mut_Y-\mu_Y}_{\infty} \leq \sum_{j = 0}^{k} a_{j}\frac{2^{-(j-1)/2}\log(5kd/\delta)}{n'} + \norm{\clamp{Y}{a_k}{+\infty}}_{\infty} \leq \frac{(k+1)\sqrt{2\ex{\norm{Y}^2_2}}\log(5kd/\delta)}{\sqrt{c}n'} + \frac{\sqrt{\ex{\norm{Y}^2_2}}}{2^{k/2}} \leq \frac{\sqrt{\ex{\norm{Y}^2_2}}\log(d/\delta)}{\sqrt{2} n}$.
\end{proof}

As a direct corollary of Proposition~\ref{Prop:ClassSubGaussian} and Theorem~\ref{Thm:MultiEstim}, we obtain the following result for estimating the mean in Euclidean norm. We prove in the next section that these bounds are optimal for all values of $n$ and $d$, up to logarithmic factors.

\begin{theorem}[\sc Multivariate estimator in Euclidean norm]
    \label{Thm:MultiEuclidean}
    There exists a quantum estimator with the following properties. Let~$X$ be a $d$-dimensional random variable with mean $\mu$ and covariance matrix $\Sigma$. Given two reals $\delta \in (0,1)$ and $n \geq \log(d/\delta)$, the estimator outputs a mean estimate $\mut$ such that
  	\[\norm{\widetilde{\mu} - \mu}_2 \leq
        \begin{cases}
      		\sqrt{\frac{\Tr(\Sigma)}{n}} + \sqrt{\frac{\norm{\Sigma} \log(1/\delta)}{n}}, & \text{if $n \leq d$}, \\[4mm]
      		\frac{\sqrt{d \Tr(\Sigma)} \log(d/\delta)}{n}, & \text{if $n > d$,}
      	\end{cases}\]
    with probability at least $1-\delta$. It uses $\wbo{n}$ quantum experiments and binary oracle queries to $X$.
\end{theorem}

\begin{proof}
  If $n \leq d$ we use any classical sub-Gaussian estimator (Proposition~\ref{Prop:ClassSubGaussian}). If $n > d$ we use the quantum estimator of Theorem~\ref{Thm:MultiEstim} and the norm inequality $\norm{\mut - \mu}_2 \leq \sqrt{d} \norm{\mut - \mu}_{\infty}$.
\end{proof}

%% file: LowerBinary.tex
Our lower bounds are based on reductions from the following composition problem, where the goal is to approximate an $N$-bit string whose entries are determined by parities over $M$ bits.

\begin{problem}[\sc $\proc{Search}^N \circ \proc{Parity}^M$]
  \label{Def:SP}
  Let $N,M \geq 1$ be two integers. Let~$\mathcal{A}_{N,M}$ denote the set of all matrices $A \in \rn^{N \times M}$ such that $\floor{N/2}$ rows have Hamming weights $\floor{M/2}$, and the other rows have Hamming weights $\floor{M/2}+1$. Define the vector $b\super{A} \in \rn^N$ such that,
  	\[b\super{A}_i = \begin{cases}
  		0, & \text{if the $i$-th row of $A$ has Hamming weight $\floor{M/2}$}, \\
  		1, & \text{if the $i$-th row of $A$ has Hamming weight $\floor{M/2}+1$},
  	\end{cases}\]
  for each $i \in [N]$. Then, the $\proc{Search}^N \circ \proc{Parity}^M$ problem consists of finding a vector $\td{b} \in \R^N$ that minimizes $\norm{\td{b} - b\super{A}}_2$ given a quantum oracle $\ket{i,j} \mapsto (-1)^{A_{i,j}} \ket{i,j}$ to $A \in \mathcal{A}_{N,M}$.
\end{problem}

We use the next lower bound for the $\proc{Search}^N \circ \proc{Parity}^M$ problem, which states that $\om{NM}$ queries are needed to approximate the vector $b\super{A}$ with small error. The proof can be easily adapted from that of \cite[Lemma 11]{vApe21c} or \cite[Lemma 5.7]{CJ21p}.

\begin{lemma}
  \label{Lem:SP}
  Let $\alpha > 1$ be a sufficiently large constant.
  Consider any quantum algorithm for the $\proc{Search}^N \circ \proc{Parity}^M$ problem that uses at most $NM/\alpha$ queries on all inputs. Then, there exists an input $A \in \mathcal{A}_{N,M}$ such that this algorithm returns a vector $\td{b}$ satisfying $\norm{\td{b} - b\super{A}}_2 \geq \om{\sqrt{N}}$ with probability at least $2/3$.
\end{lemma}

We now show that the quantum mean estimators developed in the previous sections are tight (up to logarithmic factors). For simplicity in the proof, we only consider the case of approximation in Euclidean norm. We first prove that the mean estimation problem admits no quantum advantage when the complexity parameter $n$ is smaller than the dimension. The proof works by a reduction from the $\proc{Search}^{\alpha n} \circ \proc{Parity}^1$ problem (where $\alpha$ is the constant mentioned in Lemma~\ref{Lem:SP}).

\begin{theorem}[\sc Low-precision regime]
  \label{Thm:LBLowBinary}
  Consider two integers $n,d$ such that $n \leq d/\alpha$. Fix $\sigma > 0$ and let~$\mathcal{P}_{\sigma}$ denote the set of all $d$-dimensional quantum random variables with covariance matrix~$\Sigma$ such that $\Tr(\Sigma) = \sigma^2$. Then, for any quantum estimator that uses at most $n$ binary oracle queries, there exists $X \in \mathcal{P}_{\sigma}$ such that the estimator returns a mean estimate $\mut$ of $\mu = \ex{X}$ that satisfies
    \[\norm{\mut - \mu}_2 \geq \om*{\sqrt{\frac{\Tr(\Sigma)}{n}}}\]
  with probability at least $2/3$.
\end{theorem}

\begin{proof}
  We assume for simplicity that $\alpha$ is even and~$d$ is a power of two (the other cases can be handled by simple padding arguments). Consider the partial Hadamard matrix $H \in \R^{\alpha n \times d}$ such that $H_{i,j} = \frac{1}{d} (-1)^{\inp{i}{j}}$, where $i, j \in \rn^{\log(d)}$ are written over $\log(d)$ bits. Note that $HH^{\top} = \id_{\alpha n}$ and the spectral norm of $H$ is $\norm{H} = 1$. Let $(\Omega, 2^{\Omega}, \P)$ be the probability space such that $\Omega = [\alpha n]$ and $\P(\omega) = 1/(\alpha n)$ for all $\omega \in \Omega$. For any vector $b \in \rn^{\alpha n}$ with Hamming weight $\norm{b}_1 = \alpha n/2$, define the random variable $X\super{b} : \Omega \to \R^d$ such that,
    \[X\super{b}(i) = \alpha \sigma \sqrt{\frac{n}{(\alpha^2 n-\alpha)/2}} b_i H_i\]
  where $H_i \in \R^d$ is the $i$-th row of $H$. The expectation of $X\super{b}$ is $\ex{X\super{b}} = \frac{\sigma}{\sqrt{n (\alpha^2 n-\alpha)/2}} H^{\top} b$ and the trace of its covariance matrix is $\Tr(\Sigma) = \ex{\norm{X\super{b}}_2^2} - \norm{\ex{X\super{b}}}_2^2 = \frac{2\alpha \sigma^2}{\alpha^2 n-\alpha} \norm{b}_1 - \frac{2 \sigma^2}{n (\alpha^2 n-\alpha)} \norm{b}_1 = \sigma^2$. Given any quantum estimator that uses~$n$ binary oracle queries to~$X\super{b}$ and outputs an estimate $\mut$ of $\mu = \ex{X\super{b}}$, we can transform it into an algorithm for the $\proc{Search}^{\alpha n} \circ \proc{Parity}^1$ problem that uses $n$ queries to $b$ and returns the estimate $\td{b} = \frac{\sqrt{n (\alpha^2 n-\alpha)/2}}{\sigma} H \mut$ with error $\norm{\td{b} - b}_2 \leq \frac{\sqrt{n (\alpha^2 n-\alpha)/2}}{\sigma} \norm{H(\mut - \mu)}_2 \leq \frac{\alpha n}{\sigma} \norm{\mut - \mu}_2$. Thus, by Lemma~\ref{Lem:SP}, there exists an input $b$ such that $\norm{\mut - \mu}_2 \geq \frac{\sigma}{\alpha n} \norm{\td{b} - b}_2 = \om[\big]{\frac{\sigma}{\sqrt{n}}}$.
\end{proof}

We now prove that the quantum estimator of Theorem~\ref{Thm:MultiEuclidean} is optimal in the regime where~$n$ is larger than the dimension. The proof works by a reduction from the $\proc{Search}^{d} \circ \proc{Parity}^{\alpha n/d}$ problem.

\begin{theorem}[\sc High-precision regime]
  \label{Thm:LBHighBinary}
  Consider two integers $n,d$ such that $n > d/\alpha$. Fix $\sigma > 0$ and let~$\mathcal{P}_{\sigma}$ denote the set of all $d$-dimensional quantum random variables with covariance matrix~$\Sigma$ such that $\Tr(\Sigma) = \sigma^2$. Then, for any quantum estimator that uses at most $n$ binary oracle queries, there exists $X \in \mathcal{P}_{\sigma}$ such that the estimator returns a mean estimate $\mut$ of $\mu = \ex{X}$ that satisfies
    \[\norm{\mut - \mu}_2 \geq \om*{\frac{\sqrt{d\Tr(\Sigma)}}{n}}\]
	with probability at least $2/3$.
\end{theorem}

\begin{proof}
  We assume for simplicity that $\alpha n$ is a multiple of $d$, and~$d$ is even (the other cases can be handled by simple padding arguments). Let $(\Omega, 2^{\Omega}, \P)$ be the probability space such that $\Omega = [d] \times [\alpha n/d]$ and $\P(\omega) = 1/(\alpha n)$ for all $\omega \in \Omega$. For any input $A \in \mathcal{A}_{d,\alpha n/d}$ to the $\proc{Search}^{d} \circ \proc{Parity}^{\alpha n/d}$ problem, define the random variable $X\super{A} : \Omega \to \R^d$ such that,
    \[X\super{A}(i,j) = \frac{\alpha \sigma n}{\sqrt{(\alpha n)^2-2d^2}} (-1)^{1+A_{i,j}} e_i\]
  where $e_i \in \R^d$ is the $i$-th indicator vector. The expectation of $X\super{A}$ is $\ex{X\super{A}} = \frac{2\sigma}{\sqrt{(\alpha n)^2-d^2}} b\super{A}$ and the trace of its covariance matrix is $\Tr(\Sigma) = \ex{\norm{X\super{A}}_2^2} - \norm{\ex{X\super{A}}}_2^2 = \frac{(\alpha \sigma n)^2}{(\alpha n)^2-d^2} - \frac{4\sigma^2}{(\alpha n)^2-d^2} \norm{b\super{A}}_1^2 = \sigma^2$ since $b\super{A}$ has Hamming weight~$d/2$. Given any quantum estimator that uses~$n$ binary oracle queries to~$X\super{A}$ and outputs an estimate $\mut$ of $\mu = \ex{X\super{A}}$, we can transform it into an algorithm for the $\proc{Search}^{d} \circ \proc{Parity}^{\alpha n/d}$ problem that uses $n$ queries to $A$ and returns the estimate $\td{b} = \frac{\sqrt{(\alpha n)^2-d^2}}{2\sigma} \mut$. Thus, by Lemma~\ref{Lem:SP}, there exists an input $A$ such that $\norm{\mut - \mu}_2 \geq \frac{2\sigma}{\sqrt{(\alpha n)^2-d^2}} \norm{\td{b} - b\super{A}}_2 = \om[\Big]{\frac{\sqrt{d}\sigma}{n}}$.
\end{proof}

%% file: AnalogEstimator.tex
In some cases, we might not have access to the random variable $X$ through a binary oracle~$\bora_X$, as considered in the previous part of this paper, but merely through a less powerful oracle. Several such input models arise naturally in the literature. For instance, in \cite{GAW19c}, it is shown how \textit{phase oracles} and \textit{probability oracles} arise naturally in the context of variational quantum eigensolvers, QAOA, and quantum auto-encoders. In \cite{vApe21c}, the author considers a quantum operation that prepares an unknown distribution, which we henceforth refer to as a \textit{distribution oracle}. In \cite{CJ21p}, we consider the multivariate mean estimation problem relative to all these input models.

There is one profound qualitative difference between all of these input models and the binary oracle setting considered in the previous section, which is that these input oracles in some sense preserve proximity. That is to say, if we have two random variables $X$ and $\widetilde{X}$, whose values differ by at most $\eps$ in some norm, then the operator norm difference between their respective input oracles is bounded by $O(\mathrm{poly}(\eps))$ too. This qualitatively differentiates this setting from the one considered in the previous sections, and we refer to input models satisfying this property as \textit{analog} models.

For ease of exposition, in this part of the paper we only consider the case where our random variable $X$ takes values bounded in the $d$-dimensional hypercube $[-1/4,1/4]^d$, and can be accessed through a phase oracle $\pora_X$, as defined in Definition~\ref{Def:phaseOracle}. We refer to \cite{CJ21p} for a more elaborate exposition of the other input models mentioned at the start of this section. Just like in the previous section, we assume to have access to the probability space via the oracle $U_{\P}$ of Definition~\ref{Def:qExp}.

Since its values are contained in the hypercube $[-1/4,1/4]^d$, the random variable $X$ satisfies $\mathrm{Var}[X_j] \leq 1/16$ for all $j \in [d]$, and hence $\Tr[\Sigma] \leq d/16$. This suggests that we can use the results from the previous section naively, to obtain a multivariate mean estimator in this setting as well.

There are two naive ways of approaching this. First, we can simulate a call to the binary oracle $\bora_X$, considered in the previous part of this paper, using $d$ consecutive runs of phase estimation on this phase oracle. Thus, if one only cares about the performance of the multivariate mean estimator expressed in the number of calls to $U_{\P}$, then it is clear that the same results can be obtained, that is, with $n$ calls to $U_{\P}$, one can obtain a multivariate mean estimator that finds an approximation $\widetilde{\mu}$ to $\mu$ which satisfies $\norm{\widetilde{\mu} - \mu}_{\infty} = \widetilde{O}(\sqrt{d}/n)$.

Secondly, if one only cares about the number of calls to $\pora_X$, then a little more elaborate construction also readily reduces to the binary oracle setting. Using Lemma~\ref{Lem:PrO}, one can turn the phase oracle $\pora_X$ into a probability oracle that given input $\ket{\omega}\ket{j}\ket{0}$ prepares the state $\ket{\omega}\ket{j}(\sqrt{1/2 + X(\omega)_j}\ket{1} + \sqrt{1/2 - X(\omega)_j}\ket{0})$. This operation can be combined with $U_{\P}$ and an operation that prepares the uniform superposition over all $j \in [d]$, to obtain the operator $U$ that acts as
  \[U : \ket{0}\ket{0}\ket{0} \mapsto \sum_{\omega \in \Omega} \sqrt{\P(\omega)} \ket{\omega} \otimes \frac{1}{\sqrt{d}} \sum_{j=1}^d \ket{j} \otimes \left(\sqrt{\frac12 + X(\omega)_j}\ket{1} + \sqrt{\frac12 - X(\omega)_j}\ket{0}\right).\]

Now, let $\br{\Omega} = \Omega \times [d] \times \{0,1\}$, and let $\overline{\P}(\omega,j,b) = \P(\omega)/(2d) - (-1)^b\P(\omega) X(\omega)_j/d$ be a probability measure on~$\br{\Omega}$. Furthermore, let the random variable $Y : \br{\Omega} \to \R^d$ be defined as $Y(\omega,j,b) = \pt{\ind{i = j \wedge b = 1}}_{i \in [d]}$. Then, implementing a binary oracle $\bora_Y$ is trivial, and the operation~$U$ above implements the operation~$U_{\overline{\P}}$. Finally, observe that $\ex{Y} = \vec{1}/(2d) + \ex{X}/d$. Therefore, we can find an approximation $\widetilde{\mu}_Y$ to $\mu_Y = \ex{Y}$ which satisfies $\norm{\widetilde{\mu}_Y - \mu_Y}_{\infty} = \widetilde{O}(\sqrt{d}/n')$, with~$n'$ calls to~$U$. Since $U$ requires only polylogarithmically many calls to $\pora_X$, and since $\ex{Y}$ is shrunk by a factor of $d$ compared to $\mu = \ex{X}$, we can obtain an estimate $\widetilde{\mu}$ such that $\norm{\widetilde{\mu} - \mu}_{\infty} = \widetilde{O}(d^{3/2}/n')$, with $\widetilde{O}(n')$ calls to $\pora_X$. Note that this approach also uses $\widetilde{O}(n')$ calls to $U_{\P}$, so its performance in terms of number of quantum experiments is significantly worse compared to the previous approach.


The above two considerations lead to the natural question whether it is possible to combine both approaches, i.e., whether with $n$ calls to $U_{\P}$ and $n'$ calls to $\pora_X$, it is possible to obtain an estimate $\widetilde{\mu}$ that satisfies $\norm{\widetilde{\mu} - \mu}_{\infty} = \widetilde{O}(\max\{\sqrt{d}/n, d^{3/2}/n'\})$. In this section, we show that this is indeed possible, and that one can even shave off a factor of $\sqrt{d}$ in the second branch of the maximum. This result is displayed in Theorem~\ref{Thm:AnalogEstimatorEllInfty}.

Interestingly, the performance of this algorithm can only be shown to be optimal in the regime where $n \geq d$ and $n' \geq d$, which we refer to as the high-precision regime. In the low-precision regime, i.e., when either $n < d$ or $n' < d$, we spend some extra effort to characterize the optimal precision one can obtain. If $n' < d$, then the situation turns out to be simple, since one can only attain the trivial precision $\norm{\widetilde{\mu} - \mu}_{\infty} = \widetilde{O}(1)$, which can even be achieved without making any queries at all. On the other hand, if $n < d$ and $n' \geq d$, then a small modification of the algorithm is enough to attain an optimal scaling of $\norm{\widetilde{\mu} - \mu}_{\infty} = \widetilde{O}(1/\sqrt{n})$, up to polylogarithmic factors. The low-precision regime is included in the statement of Theorem~\ref{Thm:AnalogEstimatorLowPrecision}.

As a final note, we remark that one can also use our techniques to obtain quantum mean estimators with performance guarantees in other $\ell_p$-norms, i.e., where $p \in [1,\infty)$. All precision results follow directly from simple norm conversion, i.e., they are a multiplicative factor $\Theta(d^{1/p})$ worse compared to the $\ell_{\infty}$-case. We also show this to be optimal, up to polylogarithmic factors.

\subsection{Near-optimal multivariate mean estimator in the high-precision regime}

The main technical ingredient for the phase oracle setting is presented here, which is the construction of the multivariate mean estimator in $\ell_{\infty}$-norm, in the high-precision regime. Throughout, we let $G$ be the same grid as in the previous section, i.e.,
  \[G = \set*{\frac{j}{\nn} - \frac{1}{2} + \frac{1}{2\nn} : j \in \set{0,\dots,\nn-1}}^d \subset (-1/2,1/2)^d.\]
where $m$ is a number to be determined later. When we write $u \sim G$, we mean that $u$ is taken uniformly over all elements of $G$, and again we let $\Hil_G$ be the Hilbert space spanned by mutually orthogonal computational basis states $\ket{u}$, for all $u \in G$.

We start by showing how one can use the phase oracle $\pora_X$ to compute the inner product between any vector $u \in G$ and the outcome of the random variable $X(\omega)$, and prepare the result as a phase rotation. The techniques used here are similar to those exhibited in the proof of Proposition~\ref{Prop:PhaseExp}.

\begin{lemma}[{\sc Directional phase oracle}]\label{lem:directional-phase-oracle}
	Let $d \in \N$, $\eps \in (0,1)$, $m \geq 0$, and $X$ a random variable bounded by $[-1/4,1/4]^d$. Then there exists an operator $\mathcal{L}_{X,m,\eps} : \ket{u}\ket{\omega}\qub \mapsto \ket{u}\ket{\omega}\ket{\varphi_{u,\omega}}$ acting on $\Hil_G \otimes \Hil_{\Omega} \otimes \Hila$, that can be implemented using $O((m + \log(1/\eps))\log(m/\eps))$ queries to $\pora_X$, and such that
	\[\norm*{\ket{\varphi_{u,\omega}} - e^{i\frac{m}{d}\inp{u}{X(\omega)}}\qub} \leq \eps.\]
\end{lemma}

\begin{proof}
	For every $u \in G$, we let $u^{(+)},u^{(-)} \in \R^d$ be defined as $u^{(+)}_j = \max\{u_j,0\}$ and $u^{(-)}_j = -\min\{u_j,0\}$, for all $j \in [d]$. Note $u = u^{(+)} - u^{(-)}$, and hence $\inp{u}{X(\omega)} = \inp{u^{(+)}}{X(\omega)} - \inp{u^{(-)}}{X(\omega)}$. Thus, it suffices to implement operations $\mathcal{L}_+ : \ket{u}\ket{\omega}\qub \mapsto \ket{u}\ket{\omega}\ket{\varphi_{+,\omega,u}}$ and $\mathcal{L}_- : \ket{u}\ket{\omega}\qub \mapsto \ket{u}\ket{\omega}\ket{\varphi_{-,u,\omega}}$ that satisfy
	\[\norm*{\ket{\varphi_{+,u,\omega}} - e^{i\frac{m}{d}\inp{u^{(+)}}{X(\omega)}}\qub} \leq \frac{\eps}{2}, \qquad \text{and} \qquad \norm*{\ket{\varphi_{-,u,\omega}} - e^{i\frac{m}{d}\inp{u^{(-)}}{X(\omega)}}\qub} \leq \frac{\eps}{2},\]
	because then $\mathcal{L}_{X,m,\eps} = \mathcal{L}_+\mathcal{L}_-^{\dagger}$. We now proceed to show how to implement $\mathcal{L}_+$, and omit the construction of $\mathcal{L}_-$ since it is completely analogous.

	First, we note that by adding a global phase to every call of $\pora_X$, and some local rotation on the control qubit at every controlled call of $\pora_X$, we can just as well implement the operation
	\[\ket{\omega}\ket{j} \mapsto e^{i\left(\frac12 + X(\omega)_j\right)}\ket{\omega}\ket{j}.\]

	Next, we turn this operation into a probability oracle acting on $\Hil_{\Omega} \otimes \C^d \otimes (\C^2)^{\otimes(k+1)}$, using Lemma~\ref{Lem:PrO}, with $\delta = 1/4$, and precision $\eps^2/(64m^2)$. This implements the operation $V_+ : \ket{\omega}\ket{j}\qub\ket{0} \mapsto \ket{\omega}\ket{j}(\sqrt{1-p_{\omega,j}} \qub\ket{0} + \sqrt{p_{\omega,j}} \ket{\psi}\ket{1})$, for some state $\ket{\psi}$, using $O(\log(m/\eps))$ calls to $\pora_X$, such that
	\[\left|\sqrt{p_{\omega,j}} - \sqrt{\frac12 + X(\omega)_j}\right| \leq \frac{\eps^2}{64m^2}.\]

	Next, we can implement the following operation without any queries,
	\[\ket{u}\ket{\omega}\ket{0}\qub \mapsto \ket{u}\ket{\omega} \left(\sum_{j=1}^d \sqrt{\frac{u_j^{(+)}}{d}}\ket{j} + \sqrt{1-\frac{\norm{u^{(+)}}_1}{d}}\ket{0}\right)\qub,\]
	which is a valid operation since $\norm{u^{(+)}}_1/d \leq 1/2 < 1$. By using next one call to~$V_+$, we implement the operation $W_+ : \ket{u}\ket{\omega}\qub\ket{0} \mapsto \ket{u}\ket{\omega}\left(\sqrt{p_{u,\omega}}\ket{\psi_{u,\omega}^1}\ket{1} + \sqrt{1-p_{u,\omega}}\ket{\psi_{u,\omega}^0}\ket{0}\right)$, where $\ket{\psi_{u,\omega}^0}$ and $\ket{\psi_{u,\omega}^1}$ are unit vectors, with a single call to $V_+$, such that
	\begin{align*}
		&\left|p_{u,\omega} - \frac1d\inp*{u^{(+)}}{\frac{\vec{1}}{2} + X(\omega)}\right| = \left|\sum_{j=1}^d \frac{u_j^{(+)}}{d} p_{\omega,j} - \sum_{j=1}^d \frac{u_j^{(+)}}{d} \left(\frac12 + X(\omega)_j\right)\right| \\
		& \ \leq \sum_{j=1}^d \frac{u_j^{(+)}}{d} \left|p_{\omega,j} - \left(\frac12 + X(\omega)_j\right)\right| = \sum_{j=1}^d \frac{u_j^{(+)}}{d} \left|\sqrt{p_{\omega,j}} - \sqrt{\frac12 + X(\omega)_j}\right| \cdot \left|\sqrt{p_{\omega,j}} + \sqrt{\frac12 + X(\omega)_j}\right| \\
		& \ \leq 2\frac{\norm{u^{(+)}}_1}{d} \cdot \frac{\eps^2}{64m^2} \leq \frac{\eps^2}{64m^2}.
	\end{align*}
	Furthermore, for all $a,b > 0$, we have that $|\sqrt{a} + \sqrt{b}| \geq \max\{\sqrt{a},\sqrt{b}\} = \sqrt{\max\{a,b\}} \geq \sqrt{|a - b|}$, and hence
	\[\left|\sqrt{a} - \sqrt{b}\right| = \frac{|a - b|}{\left|\sqrt{a} + \sqrt{b}\right|} \leq \frac{|a-b|}{\sqrt{|a-b|}} = \sqrt{|a-b|},\]
	which implies that
	\begin{equation}
		\label{eq:poss-amplification}
		\left|\sqrt{p_{u,\omega}} - \sqrt{\frac1d\inp*{u^{(+)}}{\frac{\vec{1}}{2} + X(\omega)}}\right| \leq \sqrt{\left|p_{u,\omega} - \frac1d\inp*{u^{(+)}}{\frac{\vec{1}}{2} + X(\omega)}\right|} \leq \frac{\eps}{8m}.
	\end{equation}

	Next, we turn the operation $W_+$ back into a phase oracle using the amplitude-to-phase conversion algorithm of Lemma~\ref{Lem:PO}, with $t = m$ and accuracy $\eps/4$, whereby we implement the unitary $\pora : \ket{u}\qub \mapsto \ket{u}\ket{\psi_{u,\omega}}$ using $O(m + \log(1/\eps))$ applications of $W_+$, which by the triangle inequality satisfies
	\begin{align*}
		&\norm*{\ket{\psi_{u,\omega}} - e^{i\frac{m}{d}\inp{u^{(+)}}{\frac{\vec{1}}{2} + X(\omega)}}\qub} \leq \norm{\ket{\psi_{u,\omega}} - e^{imp_{u,\omega}}\qub} + \left|e^{imp_{u,\omega}} - e^{i\frac{m}{d}\inp{u^{(+)}}{\frac{\vec{1}}{2} + X(\omega)}}\right| \\
		&\qquad \leq \frac{\eps}{4} + \left|m p_{u,\omega} - \frac{m}{d}\inp*{u^{(+)}}{\frac{\vec{1}}{2} + X(\omega)}\right| \\
		&\qquad = \frac{\eps}{4} + m \left|\sqrt{p_{u,\omega}} - \sqrt{\frac1d\inp*{u^{(+)}}{\frac{\vec{1}}{2} + X(\omega)}}\right| \cdot \left|\sqrt{p_{u,\omega}} + \sqrt{\frac1d\inp*{u^{(+)}}{\frac{\vec{1}}{2} + X(\omega)}}\right| \\
		&\qquad \leq \frac{\eps}{4} + 2m \cdot \frac{\eps}{8m} = \frac{\eps}{2}.
	\end{align*}

	Finally, just like at the start of this proof, we can get rid of the extra global phase $m\inp{u^{(+)}}{\vec{1}/2}$ by adding a phase gate to the control qubit whenever we call $\pora$ in a controlled manner. The resulting operation implements $\mathcal{L}_+$.

	It remains to check the number of calls to $\pora_X$ we made throughout this proof. The number of calls to $W_+$ is $O(m + \log(1/\eps))$, each of which makes $1$ call to $V_+$, which again performs $O(\log(m/\eps))$ calls to $\pora_X$. Thus, the total number of calls to $\pora_X$ amounts to $O((m + \log(1/\eps))\log(m/\eps))$. This completes the proof.
\end{proof}

One important subtlety that is a possible source for confusion is that in Equation~(\ref{eq:poss-amplification}), the thing that $p_{u,\omega}$ approximates is $\inp{u^{(+)}}{\vec{1}/2 + X(\omega)}/d$, and not $\inp{u}{\vec{1}/2 + X(\omega)}/d$. From Lemma~\ref{Lem:Truncate}, we know that the typical value of the latter would be $\norm{\vec{1}/2 + X(\omega)}_2/d \leq 1/\sqrt{d}$, and hence if we were approximating this, we could amplify away this subnormalization of $1/\sqrt{d}$ before converting everything back into a phase oracle. We cannot use this trick, however, since the typical value of $\inp{u^{(+)}}{\vec{1}/2 + X(\omega)}$ can be much bigger than $1/\sqrt{d}$. In fact, our optimality results later on in this section show that there is indeed no way to circumvent this.

Next, we show how the directional phase oracle we constructed in Lemma~\ref{lem:directional-phase-oracle} can be used to construct a directional means oracle, in a similar spirit as in Proposition~\ref{Prop:PhaseExp}. This is the objective of the following Lemma.

\begin{lemma}[{\sc Directional mean oracle constructed from phase oracle queries}]\label{Lem:AnalogDirectionalMeanOracle}
	Let $d \in \N$, $\eps,\eta \in (0,1)$, $m \geq \eps/(6\sqrt{d})$, and $X$ a random variable bounded by $[-1/4,1/4]^d$. There exists a unitary operator $\widetilde{\pora}_{X,m,\eta,\eps} : \ket{u}\ket{0} \mapsto \ket{u}\ket{\varphi_u}$ acting on $\Hil_G \otimes \Hila$ that can be implemented using $\widetilde{O}(\sqrt{d}m\log^2(1/(\eps\eta)))$ quantum experiments and $\widetilde{O}(dm\log^4(1/(\eps\eta)))$ queries to $\pora_X$, and such that
	\[\norm*{\ket{\varphi_u} - e^{i \nn\inp{u}{\ex{X}}}\qub} \leq \eps,\]
	for a fraction at least $1-\eta/2$ of all $u \in G$.
\end{lemma}

\begin{proof}
	Let $K_1,K_2 > 0$ be constants to be fixed later. By setting
	\[m' = \sqrt{\frac{d}{\log\left(\frac{144dm^2\left(\frac12 + \sqrt{d}\right)}{\eps^2\eta}\right)}}, \qquad \text{and} \qquad \eps' = \frac{1}{4K_1\log\left(\frac{m\sqrt{d}}{\eps\eta}\right) \cdot K_2\left(\frac{m\sqrt{d}}{\eta} + \log\left(\frac{1}{\eps}\right)\right)}\]
	in Lemma~\ref{lem:directional-phase-oracle}, we can implement a directional phase oracle, i.e., an operation that acts as $\ket{u}\ket{\omega}\qub \mapsto \ket{u}\ket{\omega}\ket{\chi_{u,\omega}}$, such that
	\[\norm*{\ket{\chi_{u,\omega}} - e^{i\frac{m'}{d}\inp{u}{X(\omega)}}\qub} \leq \eps',\]
	with $O((m' + \log(1/\eps')) \cdot \log(m'/\eps'))$ calls to $\pora_X$.

	Without incurring any extra overhead or error, we can also implement the operation $\ket{u}\ket{\omega}\qub \mapsto \ket{u}\ket{\omega}\ket{\psi_{u,\omega}}$, such that
	\[\norm*{\ket{\psi_{u,\omega}} - e^{i\left(\frac12 + \frac{m'}{d}\inp{u}{X(\omega)}\right)}\qub} \leq \eps',\]
	since we can always apply some $Z$-rotation with angle $1/2$ to the control qubit if we want to implement this mapping in a controlled fashion.

	Next, using Lemma~\ref{Lem:PrO} with $\delta = 1/4$ and precision $(m')^2\eps^2/(144d^2m^2)$, we can turn the above operation into a probability oracle, acting as $V_+ : \ket{u}\ket{\omega}\qub\ket{0} \mapsto \ket{u}\ket{\omega}\ket{\varphi_{u,\omega}}$, with $C_1 = O(\log(dm/(m'\eps)))$ calls to the directional phase oracle, and we let $K_1$ be the constant suppressed by the big-$O$-notation. It follows that
	\[\norm*{\ket{\varphi_{u,\omega}} - \sqrt{1-p_{u,\omega}}\qub\ket{0} - \sqrt{p_{u,\omega}}\ket{\psi}\ket{1}} \leq C_1\eps',\]
	and
	\begin{equation}
		\label{eq:approx}
		\left|\sqrt{p_{u,\omega}} - \sqrt{\frac12 + \frac{m'}{d}\inp*{u}{X(\omega)}}\right| \leq \frac{(m')^2\eps^2}{144d^2m^2}, \qquad \text{if } 4m'\left|\inp*{u}{X(\omega)}\right| \leq d.
	\end{equation}
	Let $B_{u,\omega} \in \{0,1\}$ be $1$ whenever $4m'\left|\inp*{u}{X(\omega)}\right| > d$. Then for all $\omega \in \Omega$, we have using Lemma~\ref{Lem:Truncate},
	\[\Pr_{u \sim G} \left[B_{u,\omega}\right] = \Pr_{u \sim G} \left[4m'\left|\inp{u}{X(\omega)}\right| > d\right] \leq \Pr_{u \sim G} \left[\frac{m'}{\sqrt{d}}\left|\inp{u}{X(\omega)}\right| > \norm{X(\omega)}_2\right] \leq 2e^{-\frac{2d}{(m')^2}},\]
	and using $xe^{-x} \leq e^{-x/2}$, for all $x \geq 0$, we find
	\begin{align*}
		\Pr_{u \sim G} \left[B_{u,\omega}\right] &\leq 2 \cdot \frac{(m')^2}{2d} \cdot \frac{2d}{(m')^2}e^{-\frac{2d}{(m')^2}} \leq \frac{(m')^2}{d}e^{-\frac{d}{(m')^2}} \\
		&\leq  \frac{(m')^2}{d}e^{-\log\left(\frac{144dm^2\left(\frac12 + \sqrt{d}\right)}{\eps^2\eta}\right)} = \frac{(m')^2\eps^2\eta}{144d^2m^2\left(\frac12 + \sqrt{d}\right)}.
	\end{align*}
	Thus, by averaging over all $\omega$'s, we find that
	\[\frac{1}{|G|} \sum_{u \in G} \sum_{\omega \in \Omega} \P(\omega)B_{u,\omega} \leq \frac{(m')^2\eps^2\eta}{144d^2m^2\left(\frac12 + \sqrt{d}\right)},\]
	and hence by the pigeonhole principle, we have that for at least a $(1-\eta/2)$-fraction of $u \in G$,
	\[\sum_{\omega \in \Omega}\P(\omega)B_{u,\omega} \leq \frac{(m')^2\eps^2}{72d^2m^2\left(\frac12 + \sqrt{d}\right)},\]
	i.e., for a $(1-\eta/2)$-fraction of $u \in G$, the probability of sampling an $\omega$ such that the approximation from Equation~(\ref{eq:approx}) holds is lower bounded by the right-hand side of the above equation.

	Next, we prepend $V_+$ with the mapping $\ket{u}\ket{0}\qub\ket{0} \mapsto \ket{u} \sum_{\omega \in \Omega} \sqrt{\P(\omega)}\ket{\omega}\qub \ket{0}$, which can be implemented with one quantum experiment.
	The combined operation performs the mapping $W_+ : \ket{u}\qub\ket{0} \mapsto \ket{u}\ket{\psi_u}$, with one call to $V_+$, where
	\[\norm*{\ket{\psi_u} - \sqrt{1-p_u}\ket{\psi_u^0}\ket{0} + \sqrt{p_u}\ket{\psi_u^1}\ket{1}} \leq C_1\eps',\]
	for some states $\ket{\psi_u^0}$ and $\ket{\psi_u^1}$, and such that for at least a $(1-\eta/2)$-fraction of $u \in G$,
	\begin{align*}
		&\left|p_u - \ex*{\frac12 + \frac{m'}{d} \inp{u}{X}}\right| = \left|\sum_{\omega \in \Omega} \P(\omega) p_{u,\omega} - \sum_{\omega \in \Omega} \P(\omega)\left(\frac12 + \frac{m'}{d} \inp{u}{X(\omega)}\right)\right| \\
		&\qquad \leq \sum_{\omega \in \Omega} \P(\omega) \left|p_{u,\omega} - \left(\frac12 + \frac{m'}{d}\inp{u}{X(\omega)}\right)\right| \\
		&\qquad \leq \sum_{\omega \in \Omega} \P(\omega)B_{u,\omega} \cdot \left[\frac12 + m'\right] + \sum_{\omega \in \Omega} \P(\omega)(1-B_{u,\omega}) \left|\sqrt{p_{u,\omega}} - \sqrt{\frac12 + \frac{m'}{d}\inp{u}{X(\omega)}}\right| \cdot 2\\
		&\qquad \leq \frac{(m')^2\eps^2\left(\frac12 + m'\right)}{72d^2m^2\left(\frac12+\sqrt{d}\right)} + 2 \cdot \frac{(m')^2\eps^2}{144d^2m^2} \leq \frac{(m')^2\eps^2}{36d^2m^2},
	\end{align*}
	where in the last line we used that $m' \leq \sqrt{d}$. We conclude that for at least a $(1-\eta/2)$-fraction of $u \in G$,
	\[\left|\sqrt{p_u} - \sqrt{\frac12 + \frac{m'}{d}\inp{u}{\ex{X}}}\right| \leq \sqrt{\left|p_u - \ex*{\frac12 + \frac{m'}{d} \inp{u}{X}}\right|} \leq \frac{m'\eps}{6dm},\]
	following a same argument as in the proof of the previous lemma.

	Then, we convert the resulting operation $W_+$ back into a phase oracle, using Lemma~\ref{Lem:PO} with precision $\eps/4$, and $t = dm/m'$. Then, the resulting operation performs $\ket{u}\ket{0} \mapsto \ket{u}\ket{\chi_u}$ with $C_2 = O(dm/m' + \log(1/\eps))$ calls to $W_+$, and we let $K_2$ be the constant suppressed by the big-$O$-notation. Then we find, by the triangle inequality, that for at least a $(1-\eta/2)$-fraction of $u \in G$,
	\begin{align*}
		&\norm*{\ket{\chi_u} - e^{i\frac{dm}{m'}\left(\frac12 + \frac{m'}{d}\inp{u}{\ex{X}}\right)}\qub} \\
		&\qquad \leq C_1C_2\eps' + \norm*{\ket{\chi_u} - e^{i\frac{dm}{m'}p_u}\qub} + \left|e^{i\frac{dm}{m'}p_u} - e^{i\frac{dm}{m'}\left(\frac12 + \frac{m'}{d}\inp{u}{\ex{X}}\right)}\right| \\
		&\qquad \leq \frac{\eps}{4} + \frac{\eps}{4} + \left|\frac{dm}{m'}p_u - \frac{dm}{m'}\left(\frac12 + \frac{m'}{d}\inp{u}{\ex{X}}\right)\right| \\
		&\qquad = \frac{\eps}{2} + \frac{dm}{m'}\left|\sqrt{p_u} - \sqrt{\frac{m'}{d}\inp{u}{\ex{X}}}\right| \cdot \left|\sqrt{p_u} + \sqrt{\frac{m'}{d}\inp{u}{\ex{X}}}\right| \\
		&\qquad \leq \frac{\eps}{2} + \frac{dm}{m'} \cdot \frac{m'\eps}{6dm} \cdot \left(2 + \frac{m'\eps}{6dm}\right) \leq \frac{\eps}{2} + \frac{\eps}{6} \cdot 3 = \eps,
	\end{align*}
	where we used that $m' \leq \sqrt{d}$ and $m \geq \eps/(6\sqrt{d})$ in the last inequality.

	Finally, we can remove the unnecessary global phase $dm/(2m')$ by applying some $Z$-rotation on any control qubit when we call the above operation in a controlled manner, which does not incur any additional overhead in terms of the number of queries or error. Thus, we have shown how to implement $\widetilde{P}_{X,m,\eta,\eps}$.

	It remains to check how many quantum experiments and queries to $\pora_X$ we have performed throughout its construction. Multiplying  the complexities appearing earlier in this proof together results in
	\[O\left(\left(\frac{dm}{m'} + \log\left(\frac{1}{\eps}\right)\right) \cdot \log\left(\frac{dm}{m'\eps}\right)\right) = \widetilde{O}\left(\sqrt{d}m\log^2\left(\frac{1}{\eps\eta}\right)\right)\]
	quantum experiments, and
	\[O\left(\left(\frac{dm}{m'} + \log\left(\frac{1}{\eps}\right)\right) \cdot \log\left(\frac{dm}{m'\eps}\right) \cdot \left(m' + \log\left(\frac{1}{\eps'}\right)\right) \cdot \log\left(\frac{m'}{\eps'}\right)\right),\]
	queries to $\pora_X$, which after substitution of $m'$ and $\eps'$ can be upper bounded by
	\[\widetilde{O}\left(dm\log^4\left(\frac{1}{\eps\eta}\right)\right).\]
	This completes the proof.
\end{proof}

Now, we are ready to put everything together, and provide a full construction of a multivariate quantum mean estimator using phase oracles. The core idea is to use the bounded multivariate estimator from Algorithm~\ref{Alg:MultiBoundedEst}, with slightly tweaked constants to accommodate for the slight differences in the guarantees we have on the precision of the directional means oracle. The full algorithm is presented in Algorithm~\ref{Alg:AnalogEstimatorEllInfty}.

\algobox{Alg:AnalogEstimatorEllInfty}{Multivariate mean estimator with phase oracles, $\proc{QPhase}_d\pt{X,n,n',\delta}$.}{
\begin{enumerate}[leftmargin=*]
	\item Set $k = \left\lfloor\min\left\{n,\frac{n'}{\sqrt{d}}\right\}\right\rfloor$, $\eta = \frac{1}{288}$ and $\nn = 2^{\ceil[\big]{\log\pt[\big]{\frac{8\pi k}{\sqrt{d} \log(d/\delta)}}}}$.
	\item For $\ell = 1,\dots,\ceil{18\log(d/\delta)}$:
	\begin{enumerate}
		\item Compute the uniform superposition $\ket{G} := \frac{1}{\nn^{d/2}} \sum_{u \in G} \ket{u}$ over $G$.\label{Step:startAn}
		\item Compute the state $\ket{\psi} := \td{\pora}_{X,\nn,\eta,\eps} \ket{G}\qub \in \Hil_G \otimes \Hila$, where $\td{\pora}_{X,\nn,\eta,\eps}$ is the directional means oracle constructed in Lemma~\ref{Lem:AnalogDirectionalMeanOracle} with $\eps = 1/(12\sqrt{2})$.
		\item Compute the state $\ket{\phi} := \pt{\mathrm{QFT}^{-1}_G \otimes \id_{\mathrm{aux}}}\ket{\psi}$ where the unitary  $\mathrm{QFT}_G : \ket{u} \mapsto \frac{1}{\nn^{d/2}} \sum_{v \in G} e^{2i\pi \nn \inp{u}{v}} \ket{v}$ is the quantum Fourier transform over~$G$.
		\item Measure the $\Hil_G$ register of $\ket{\phi}$ in the computational basis and let $\td{v}\super{\ell} \in G$ denote the obtained result. Set $\td{\mu}\super{\ell} = 2\pi\td{v}\super{\ell}$.\label{Step:measAn}
	\end{enumerate}
	\item Output the coordinate-wise median, $\td{\mu} = \median\pt{\td{\mu}\super{1},\dots,\td{\mu}\super{{\ceil{18\log(d/\delta)}}}}$.
\end{enumerate}
}

\begin{theorem}[{\sc High-precision multivariate mean estimator with phase oracles}]\label{Thm:AnalogEstimatorEllInfty}
	Let $d \in \N$, $\delta \in (0,1)$, $n \geq \log(d/\delta)$, $n' \geq \sqrt{d}\log(d/\delta)$, and $X$ a random variable bounded by $[-1/4,1/4]^d$, with $\mu = \ex{X}$. Then the multivariate mean estimator with phase oracles, $\proc{QPhase}_d(X,n,\delta)$ (Algorithm~\ref{Alg:AnalogEstimatorEllInfty}), finds an approximation to the mean, $\widetilde{\mu}$, that with probability at least $1-\delta$ satisfies
	\[\norm{\widetilde{\mu} - \mu}_{\infty} \leq \max\left\{\frac{\sqrt{d}}{n}, \frac{d}{n'}\right\} \cdot \log\left(\frac{d}{\delta}\right),\]
	with $\widetilde{O}(n)$ calls to $U_{\P}$ and $\widetilde{O}(n')$ calls to $\pora_X$.
\end{theorem}

\begin{proof}
	We follow the general proof strategy from Theorem~\ref{Thm:MultiBoundedEst}, and let
	\[\ket{\psi'} = \frac{1}{\sqrt{m^d}} \sum_{u \in G} e^{im\inp{u}{\ex{X}}} \ket{u}\qub.\]
	From the performance guarantee on $\widetilde{\pora}_{X,m,\eta,\eps}$ that we proved in Lemma~\ref{Lem:AnalogDirectionalMeanOracle}, we find that $\norm{\ket{\psi} - \ket{\psi'}}^2 \leq \sum_{u \in G} \norm{\ket{\varphi_u} - e^{im\inp{u}{\ex{X}}}\qub}^2/m \leq \eta + \eps^2 \leq 1/144$.

	We now analyze the remainder of the algorithm as if the state $\ket{\psi'}$ was prepared, instead of $\ket{\psi}$. Let $\widetilde{v} \in G$ be the outcome of the measurement performed in step~\ref{Step:measAn} of the algorithm. Since $\norm{\ex{X}}_{\infty} \leq 1/4 < 2\pi/3$, by the standard analysis of the phase estimation algorithm, as can for instance be found in Equation~(5.34) in \cite{NC11b}, for every $j \in [d]$ we have $|\widetilde{v}_j - \ex{X}/(2\pi)| \leq 4/m$ with probability at least $5/6$. If we now factor in that we start with $\ket{\psi}$ rather than $\ket{\psi'}$, the probability goes down from $5/6$ to $2/3$. Finally, from the Chernoff bound, it follows that $\norm{\widetilde{\mu} - \ex{X}}_{\infty} \leq 8\pi/m \leq \sqrt{d}\log(d/\delta)/k$ with probability at least $1-\delta$, from which the claim follows.

	It remains to analyze the query complexity claims. We make $O(\log(d/\delta))$ calls to the directional means oracle from Lemma~\ref{Lem:AnalogDirectionalMeanOracle}, from which we find that the number of quantum experiments is $\widetilde{O}(\sqrt{d}m\log^2(1/(\eps\eta))\log(d/\delta)) = \widetilde{O}\left(k\right) = \widetilde{O}(n)$, and similarly the number of calls to $\pora_X$ is $\widetilde{O}(dm\log^4(1/(\eps\eta))\log(d/\delta)) = \widetilde{O}(k\sqrt{d}) = \widetilde{O}(n')$, completing the proof.
\end{proof}

Later on in this section, we find corresponding lower bounds on the precision that scale as $\Omega(\max\{\sqrt{d}/n, d/n'\})$, implying that the performance guarantee we obtain here is optimal up to polylogarithmic factors. However, this lower bound only holds in the regime where both $n \geq d$ and $n' \geq d$, and surprisingly it turns out that one can do better in the case where either $n < d$ or $n' < d$. We show this in the next section.

\subsection{Near-optimal multivariate mean estimator in the low-precision regime}

It turns out that the performance of Algorithm~\ref{Alg:AnalogEstimatorEllInfty} is only optimal in the regime where $n \geq d$ and $n' \geq d$. In this section, we take a look at the regime where we have very few calls to the input oracles to spend, more specifically where $n < d$ or $n' < d$. We refer to this regime as the low-precision regime.

If $n' < d$, then the performance bound on Algorithm~\ref{Alg:AnalogEstimatorEllInfty} becomes at least $d/n' > 1$. However, we know a priori that $\ex{X}$ is contained in $[-1/4,1/4]^d$, so if we just output the all-zeros vector, we will do better than what Theorem~\ref{Thm:AnalogEstimatorEllInfty} suggests. Moreover, we will show in the next section that this is actually optimal, i.e., if one has less than $d$ queries to $\pora_X$ to spend, one might as well just output the all-zeros vector, since there is nothing one can do that will provably result in a significantly better estimate.

This leaves the regime where $n < d$ and $n' \geq d$, and it turns out that in this regime there is indeed a non-trivial approach that beats the complexity obtained by Algorithm~\ref{Alg:AnalogEstimatorEllInfty}. The modification is very simple -- one just samples from the probability space $n$ times, and then runs Algorithm~\ref{Alg:AnalogEstimatorEllInfty} with the empirical distribution. The algorithm is presented in Algorithm~\ref{Alg:LowPrecisionAnalogEstimator}, and the performance guarantees are presented in Theorem~\ref{Thm:AnalogEstimatorLowPrecision}.

\algobox{Alg:LowPrecisionAnalogEstimator}{Low-precision multivariate mean estimator, $\proc{QLowPrecPhase}_d\pt{X,n,n',\delta}$.}{
\begin{enumerate}[leftmargin=*]
	\item Set $k' = \lfloor 2n/\log(d/\delta)\rfloor$.
	\item For $\ell = 1,\dots,\ceil{32\log(d/\delta)}$:
	\begin{enumerate}
		\item Obtain samples $\omega^{(1)}, \dots, \omega^{(k')}$ from the probability space, and let $\overline{\P}$ be the empirical distribution based on the observed samples.\label{Step:sampAn}
		\item Run steps~\ref{Step:startAn} to~\ref{Step:measAn} of Algorithm~\ref{Alg:AnalogEstimatorEllInfty}, with $k = 2n'/\sqrt{d}$, all other parameters chosen identically, and the quantum experiment oracle $U_{\overline{\P}}$ constructed from the observed samples. Denote the outcome by $\widetilde{\mu}^{(\ell)}$.
	\end{enumerate}
	\item Output the coordinate-wise median, $\td{\mu} = \median\pt{\td{\mu}\super{1},\dots,\td{\mu}\super{{\ceil{18\log(d/\delta)}}}}$.
\end{enumerate}
}

\begin{theorem}[{\sc Low-precision analog mean estimator}]\label{Thm:AnalogEstimatorLowPrecision}
	Let $d \in \N$, $\delta \in (0,1)$, $n \geq \log(d/\delta)$, $n' \geq \sqrt{d}\log(d/\delta)$, and $X$ a random variable with values contained in $[-1/4,1/4]^d$, with $\mu = \ex{X}$. Then, $\proc{QLowPrecPhase}_d(X,n,n',\delta)$ (Algorithm~\ref{Alg:LowPrecisionAnalogEstimator}) finds an approximation to the mean, $\widetilde{\mu}$, that with probability at least $1-\delta$ satisfies
	\[\norm{\widetilde{\mu} - \mu}_{\infty} \leq \max\left\{\frac{1}{\sqrt{n}}, \frac{d}{n'}\right\} \cdot \log\left(\frac{d}{\delta}\right),\]
	with $\widetilde{O}(n)$ calls to $U_{\P}$, and $\widetilde{O}(n')$ calls to $\pora_X$.
\end{theorem}

\begin{proof}
	Since we only call $U_{\P}$ in step~\ref{Step:sampAn}, it is clear we perform a total of $O(k'\log(d/\delta)) = O(n)$ quantum experiments. Similarly, the number of calls to the phase oracle $\pora_X$ is twice that in a run of Algorithm~\ref{Alg:AnalogEstimatorEllInfty} with $n > n'/\sqrt{d}$, from which we readily deduce that it is indeed $\widetilde{O}(n')$.

	It remains to check the precision guarantee. To that end, let
	\[\overline{\P}_{\omega} = \frac{|\{j \in [k'] : \omega^{(j)} = \omega\}|}{k'},\]
	and let $\overline{\mu} = \ex_{\overline{\P}}{X}$, i.e., the mean of $X$ under this empirical probability distribution. Note that $\overline{\mu}$ itself is also a random variable, since it depends on the $\omega$'s observed in the first stage of the algorithm. From the analysis in Theorem~\ref{Thm:AnalogEstimatorEllInfty}, we find that for all $j \in [d]$ and $\ell \in [\ceil{32\log(d/\delta)}]$, $|\widetilde{\mu}_j - \overline{\mu}_j| \leq d\log(d/\delta)/(2n')$ with probability at least $2/3$. Thus, it remains to show that with high probability $\overline{\mu}$ approximates $\mu = \ex{X}$ well, that is, it remains to show that for all $j \in [d]$, $|\overline{\mu}_j - \mu_j| \leq \log(d/\delta)/(2\sqrt{n})$ with high probability.

	To that end, observe that for all $j \in [d]$,
	\[\left|\ex{\overline{\mu}_j} - \mu_j\right| = \left|\sum_{\omega \in \Omega} \overline{\P}_{\omega} X_j(\omega) - \P(\omega)X_j(\omega)\right| \leq \frac14\left|\sum_{\omega \in \Omega} \overline{\P}_{\omega} - \P(\omega)\right|,\]
	and also, using that $\ex{\overline{\P}_{\omega}} = \P(\omega)$,
	\begin{align*}
		\ex*{\left(\sum_{\omega \in \Omega} \overline{\P}_{\omega} - \P(\omega)\right)^2} &= \mathrm{Var}\left[\sum_{\omega \in \Omega} \overline{\P}_{\omega} - \P(\omega)\right] = \sum_{\omega \in \Omega}\mathrm{Var}\left[\overline{\P}_{\omega}\right] \\
		&= \sum_{\omega \in \Omega} \frac{\P(\omega)(1-\P(\omega))}{k'} \leq \sum_{\omega \in \Omega} \frac{\P(\omega)}{k'} = \frac{1}{k'}.
	\end{align*}
	Therefore, by Markov's inequality, for all $j \in [d]$,
	\[\P\left[\left|\overline{\mu}_j - \mu_j\right| > \frac{\log(d/\delta)}{2\sqrt{n}}\right] \leq \P\left[\left|\sum_{\omega \in \Omega} \overline{\P}_{\omega} - \P(\omega)\right| > \frac{4}{\sqrt{k'}}\right] = \P\left[\left|\sum_{\omega \in \Omega} \overline{\P}_{\omega} - \P(\omega)\right|^2 > \frac{16}{k'}\right] \leq \frac{1}{16}.\]
	Thus, by the triangle inequality, for every $j \in [d]$ and $\ell \in [\ceil{32\log(d/\delta)}]$, with probability at least $2/3 \cdot 15/16 = 5/8$ we have that $|\widetilde{\mu}^{(\ell)}_j - \mu_j| \leq |\widetilde{\mu}^{(\ell)}_j - \overline{\mu}_j| + |\overline{\mu}_j - \mu_j| \leq (d/(2n') + 1/(2\sqrt{n})) \cdot \log(d/\delta) \leq \max\{1/\sqrt{n}, d/n'\} \cdot \log(d/\delta)$. Finally, it follows from the Chernoff bound that after $\ceil{32\log(d/\delta)}$ iterations, we obtain that $\norm{\widetilde{\mu} - \mu}_{\infty} \leq \max\{1/\sqrt{n}, d/n'\} \cdot \log(d/\delta)$, completing the proof.
\end{proof}

We have now described all algorithms. For convenience, we aggregate all algorithmic results in one self-contained statement.

\begin{theorem}
  \label{Thm:phOracle}
	Let $d,n,n' \in \N$, $\delta \in (0,1)$, and $X$ a random variable with values contained in $[-1/4,1/4]^d$, with $\mu = \ex{X}$. Then, we can find an approximation to the mean, $\widetilde{\mu}$, that with probability at least $1-\delta$ satisfies
	\[\norm{\widetilde{\mu} - \mu}_{\infty} \leq \begin{cases}
		1, & \text{if $n' < d$ or $n < \log(d/\delta)$}, \\[3mm]
		\max\left\{\frac{1}{\sqrt{n}}, \frac{d}{n'}\right\} \cdot \log\left(\frac{d}{\delta}\right), & \text{if $n' \geq d$ and $\log(d/\delta) \leq n < d$}, \\[4mm]
		\max\left\{\frac{\sqrt{d}}{n}, \frac{d}{n'}\right\} \cdot \log\left(\frac{d}{\delta}\right), & \text{if $n' \geq d$ and $n \geq d$}.
	\end{cases}\]
	with $\widetilde{O}(n)$ calls to $U_{\P}$, and $\widetilde{O}(n')$ calls to $\pora_X$. Furthermore, for all $p \in [1,\infty)$, we obtain the same performance guarantees on $\norm{\widetilde{\mu} - \mu}_p$, but multiplied with $d^{1/p}$.
\end{theorem}

\begin{proof}
	All statements are already present in Theorem~\ref{Thm:AnalogEstimatorEllInfty} and Theorem~\ref{Thm:AnalogEstimatorLowPrecision}, except for the case where $p \in [1,\infty)$, which follows from the norm inequality $\norm{\mathbf{x}}_p \leq d^{1/p}\norm{\mathbf{x}}_{\infty}$, for all $\mathbf{x} \in \R^d$.
\end{proof}

There exist at least three ways in which this result can be generalized. For instance, one can ask the question how many queries are required if instead of assuming that the random variable is bounded by $[-1/4,1/4]^d$, it is instead bounded by some $\ell_q$-ball of radius $1/4$, with $q \in [1,\infty)$. One can also wonder how many queries are required when one wants to obtain an approximation of $O\ex{X}$, where $O$ is some $d$-dimensional rotation matrix, i.e., a matrix that satisfies $O^TO = I$. Finally, if one has access to $X$ through some other oracle than a phase oracle, then one can also wonder how many queries to such an oracle are required to solve the multivariate mean estimation problem. These questions are all addressed in \cite{CJ21p}, albeit only in the high-precision regime. It is an interesting direction for further research to tightly characterize the query complexities of these problems in the low-precision setting as well.

%% file: LowerAnalog.tex
We now turn our attention to proving lower bounds on the number of queries required to the input oracles. We first focus our attention on the high-precision regime, and will show later on that the lower bounds in the low-precision regimes follow via some reduction from those that we prove in the high-precision regime.

As is customary with lower bounding, we would like to embed a problem whose hardness has already been shown before in the setting we consider here, in order to conclude that this problem must be at least as hard to solve. We start by considering the problem of recovering a constant fraction of the bits in a bit string when we are given access to it by means of a fractional phase oracle.

\begin{lemma}\label{Lem:StringRecovery}
	Let $\eps \in (0,\pi]$, $d \in \N$, and suppose that we have access to a bit string $\mathbf{b} \in \{0,1\}^d$ through controlled calls to a fractional phase oracle $\fora_{\eps} : \ket{j} \mapsto e^{i\eps b_j}\ket{j}$. Then, in order to find a bit string $\widetilde{\mathbf{b}} \in \{0,1\}^d$ such that $\norm{\widetilde{\mathbf{b}} - \mathbf{b}}_1 \leq d/4$ with probability at least $2/3$, we must make at least $\Omega(d/\eps)$ calls to $\fora_{\eps}$.
\end{lemma}

\begin{proof}
	First, we argue that it is sufficient to consider the case where $\eps = \pi$. Indeed, in general, the query complexity of any problem is increased by a multiplicative factor of $\Theta(1/\eps)$, when one changes the input model from a regular phase oracle $\fora_{\pi}$ to a fractional phase oracle $\fora_{\eps}$. In Appendix B of \cite{LMR11c}, this is proven for problems that can be phrased as a binary function. However, since the problem we consider here does not have a unique correct output on every given input, we must combine their technique with the general adversary bound for relations, as derived by \cite{Bel15p}, to arrive at the desired result. More details can be found in \cite{CJ21p}.

	Thus, it remains to focus on the case where $\eps = \pi$. Suppose that we have an algorithm $\mathcal{A}$ that finds a bit string $\widetilde{\mathbf{b}} \in \{0,1\}^d$ such that with probability at least $2/3$, we have $\norm{\widetilde{\mathbf{b}} - \mathbf{b}}_1 \leq d/4$, i.e., $\widetilde{\mathbf{b}}$ and $\mathbf{b}$ differ in at most $d/4$ bits. Then, we can let $\mathcal{B}$ be the quantum algorithm that first runs $\mathcal{A}$ to obtain such a bit string $\widetilde{\mathbf{b}}$, and then selects uniformly at random a bit string $\overline{\mathbf{b}} \in \{0,1\}^d$ that satisfies $\norm{\overline{\mathbf{b}} - \widetilde{\mathbf{b}}}_1 \leq d/4$. We have $M = \sum_{t=0}^{\lfloor d/4\rfloor} \binom{d}{t}$ possible choices, which implies that the probability of this algorithm outputting $\mathbf{b}$ exactly is lower bounded by $2/3 \cdot 1/M$. By the information theoretic lower bound, i.e., Equation~(4) in \cite{FGGS99j}, the number of queries to $\fora_{\pi}$, performed by $\mathcal{B}$ and hence also by $\mathcal{A}$, denoted by $Q$, satisfies
	\[2^d \leq \frac32 \cdot \sum_{t=0}^{\left\lfloor \frac{d}{4} \right\rfloor} \binom{d}{t} \sum_{t=0}^Q \binom{d}{t} \leq \frac32 \cdot 2^{d\left(H\left(\frac14\right) + H\left(\frac{Q}{d}\right)\right)},\]
	where in the final inequality, we used a well-known approximation sums of binomial coefficients, as proven for instance in Lemma~16.19 in \cite{FG06b}, and $H(x) = -x\log(x) - (1-x)\log(1-x)$ is the binary entropy function. Taking logarithms on both sides yields that $H(Q/d) \geq 1 - H(1/4) + o(1)$, which implies that $Q = \Omega(d)$, completing the proof.
\end{proof}

The hardness of this problem can be used as a black box to show the high-precision lower bound on the precision we can attain, expressed in the number of calls to $\pora_X$, as is shown in the theorem below.

\begin{theorem}
	\label{Thm:LBPOanalog}
	Let $d \in \N$, $n' \geq d$, and suppose that we have a quantum algorithm that finds an approximation $\widetilde{\mu}$ of the mean $\mu$ of any random variable with values contained in $[-1/4,1/4]^d$, using $n'$ queries to $\pora_X$. Then, there exist instances in which case the error between $\widetilde{\mu}$ and $\mu$ satisfies
	\[\norm{\widetilde{\mu} - \mu}_1 = \Omega\left(\frac{d^2}{n'}\right),\]
	with probability at least $2/3$.
\end{theorem}

\begin{proof}
	Let $\eps = d/n'$, and $\mathbf{b} \in \{0,1\}^d$ a bit string, that we can access through controlled calls to a fractional phase oracle $\fora_{\eps} : \ket{j} \mapsto e^{i\eps b_j}\ket{j}$. Then, we know from Lemma~\ref{Lem:StringRecovery} that it takes $\Omega(d/\eps) = \Omega(n')$ calls to $\fora_{\eps}$ to find a bit string $\widetilde{\mathbf{b}} \in \{0,1\}^d$ such that $\norm{\widetilde{\mathbf{b}} - \mathbf{b}}_1 \leq d/4$.

	Now, let $\Omega = \{0\}$, with $\P(0) = 1$, which implies that $U_{\P} = I$. Let the random variable $X : \Omega \to \R^d$ be defined as $X(0)_j = \eps b_j$. This implies that $\pora_X : \ket{0}\ket{j} \mapsto e^{i\eps b_j}\ket{0}\ket{j}$, and hence $\pora_X$ can be implemented with one call to $\fora_{\eps}$. Furthermore, we have $\mu = \ex{X} = \eps\mathbf{b}$, and hence if we can find an approximation $\widetilde{\mu}$ to the mean satisfying $\norm{\widetilde{\mu} - \mu} \leq d^2/(8n')$, then we have
	\[\min_{\widetilde{\mathbf{b}} \in \{0,1\}^d} \norm*{\frac{\widetilde{\mu}}{\eps} - \widetilde{\mathbf{b}}}_1 = \frac{1}{\eps} \min_{\widetilde{\mathbf{b}} \in \{0,1\}^d} \norm{\widetilde{\mu} - \eps\widetilde{\mathbf{b}}}_1 \leq \frac{1}{\eps} \norm{\widetilde{\mu} - \mu}_1,\]
	and hence, if we take $\widetilde{\mathbf{b}}$ to be the bit string attaining the minimum in the left-hand side, we find by the triangle inequality
	\[\norm{\widetilde{\mathbf{b}} - \mathbf{b}}_1 \leq \norm*{\frac{\widetilde{\mu}}{\eps} - \widetilde{\mathbf{b}}}_1 + \norm*{\frac{\widetilde{\mu}}{\eps} - \mathbf{b}}_1 \leq \frac{1}{\eps} \norm{\widetilde{\mu} - \mu}_1 + \frac{1}{\eps} \norm{\widetilde{\mu} - \mu}_1 = \frac{2}{\eps}\norm{\widetilde{\mu} - \mu}_1 \leq \frac{d}{4}.\]
	But we know that this takes at least $\Omega(n')$ calls to $\fora_{\eps}$, and hence obtaining an estimate $\widetilde{\mu}$ that satisfies $\norm{\widetilde{\mu} - \mu}_1 \leq d^2/(8n')$ requires at least $\Omega(n')$ queries to $\pora_X$ as well. Thus, if we only have $n'$ queries to spend, there must exist instances such that $\norm{\widetilde{\mu} - \mu}_1 = \Omega(d^2/n')$ with high probability, completing the proof.
\end{proof}

In order to give a similar lower bound in terms of the number of quantum experiments, we need to subtly change the problem to finding a vector $\widetilde{\mathbf{b}} \in \{0,1\}^d$ such that $\norm{H(\widetilde{\mathbf{b}} - \mathbf{b})}_1 \leq d/4$, where $H$ is a $d$-dimensional normalized Hadamard matrix, i.e., the entries of $H$ are all $\pm 1/\sqrt{d}$, and $H^TH = I$. This problem is quite different in nature compared to the problem considered in Lemma~\ref{Lem:StringRecovery}, since the number of bits in which $\mathbf{b}$ and $\widetilde{\mathbf{b}}$ differ seems to no longer tells us anything useful about whether this condition is met. For instance, if $\mathbf{b} = \mathbf{0}$ and $\widetilde{\mathbf{b}} = \mathbf{1}$, i.e., they differ in \textit{all} bits, then $H(\widetilde{\mathbf{b}} - \mathbf{b}) = \sqrt{d}\mathbf{e}_1$, and so the condition is met as long as $d$ is large enough. Suprisingly, however, we are able to show that this problem is equally hard as the non-rotated problem up to constants, i.e., it still takes $\Omega(d/\eps)$ calls to $\fora_{\eps}$ to find a $\widetilde{\mathbf{b}}$ that satisfies this rotated condition. The details can be found in the lemma below.

\begin{lemma}\label{Lem:RotatedStringRecovery}
	Let $\eps \in (0,\pi]$, $d \in \N$ a power of $2$, and let $H$ be a $d$-dimensional Hadamard matrix, i.e., all entries of $H$ are $\pm 1/\sqrt{d}$, and $H^TH = I$. Suppose we have access to a bit string $\mathbf{b} \in \{0,1\}^d$ by means of a fractional phase oracle $\fora_{\eps} : \ket{j} \mapsto e^{i\eps b_j}\ket{j}$. In order to find a bit string $\widetilde{\mathbf{b}} \in \{0,1\}^d$ such that $\norm{H(\widetilde{\mathbf{b}} - \mathbf{b})}_1$, the number of calls to $\fora_{\eps}$ satisfies $\Omega(d/\eps)$.
\end{lemma}

\begin{proof}
	Similarly as in the proof of Lemma~\ref{Lem:StringRecovery}, it suffices to consider the case where $\eps = \pi$. Suppose that we have a quantum algorithm $\mathcal{A}$ that finds a bit string $\widetilde{\mathbf{b}} \in \{0,1\}^d$ satisfying the condition posed in the statement of the lemma with probability at least $2/3$. Using heavy-duty tools from statistics, it is shown in Appendix B of \cite{CJ21p} that
	\[\Pr_{\overline{\mathbf{b}} \sim \{0,1\}^d}\left[\norm{H(\widetilde{\mathbf{b}} - \overline{\mathbf{b}})}_1 \leq \frac{d}{4}\right] \leq 2^{-Cd}, \qquad \text{with} \qquad C = \frac{\log(e)}{2}\left(\frac{1}{2\sqrt{2}} - \frac14\right)^2 \in (0,1).\]
	Thus, if know $\widetilde{\mathbf{b}} \in \{0,1\}^d$ and that $\norm{H(\widetilde{\mathbf{b}} - \mathbf{b})}_1 \leq d/4$, then there are less than $2^{(1-C)d}$ possible choices left for $\mathbf{b}$. Thus, the algorithm $\mathcal{B}$, that first runs $\mathcal{A}$ to obtain a vector $\widetilde{\mathbf{b}}$ such that $\norm{H(\widetilde{\mathbf{b}} - \mathbf{b})}_1 \leq d/4$, and subsequently takes any $\overline{\mathbf{b}}$ that satisfies $\norm{H(\widetilde{\mathbf{b}} - \overline{\mathbf{b}})}_1 \leq d/4$ uniformly at random, will recover $\mathbf{b}$ exactly with probability lower bounded by $2/3 \cdot 2^{(C-1)d}$. Analogously as in the proof of Lemma~\ref{Lem:StringRecovery}, this implies that
	\[2^d \leq \frac32 \cdot 2^{(1-C)d} \cdot \sum_{t=0}^Q \binom{d}{t} \leq \frac32 \cdot 2^{(1-C)d + dH\left(\frac{Q}{d}\right)},\]
	and taking the logarithm on both sides implies that $H(Q/d) \geq C + o(1)$, which in turn implies that $Q = \Omega(d)$, completing the proof.
\end{proof}

We now show how the hardness of problem considered in the previous lemma can be used to lower bound the query complexity in the mean estimation problem.

\begin{theorem}
	\label{Thm:LBUPanalog}
	Let $d \in \N$, $n \geq d$, and suppose that we have a quantum algorithm that finds an approximation $\widetilde{\mu}$ to the mean $\mu$ of any random variable with values contained in $[-1/4,1/4]^d$, using $n$ queries to $U_{\P}$. Then, there exist instances in which case the error between $\widetilde{\mu}$ and $\mu$ satisfies
	\[\norm{\widetilde{\mu} - \mu}_1 = \Omega\left(\frac{d^{\frac32}}{n}\right),\]
	with probability at least $2/3$.
\end{theorem}

\begin{proof}
	Let $d'$ be the biggest power of $2$ below or equal to $d$. Let $\eps = d'/n$, $\eps' = \arcsin(\eps)$, and suppose that we have access to some hidden bit string $\mathbf{b} \in \{0,1\}^{d'}$ by means of controlled calls to a fractional phase oracle $\fora_{\eps'} : \ket{j} \mapsto e^{i\eps'b_j}\ket{j}$. We know from Lemma~\ref{Lem:RotatedStringRecovery} that it takes $\Omega(d'/\eps') = \Omega(n)$ calls to find a bit string $\widetilde{\mathbf{b}}$ such that $\norm{H(\widetilde{\mathbf{b}} - \mathbf{b})}_1 \leq d'/4$.

	Now, let $\Omega = [d'] \times \{0,1\}$, and for every $\mathbf{b} \in \{0,1\}^{d'}$, let the probability measure $\P_{\mathbf{b}}$ on $\Omega$ be defined as
	\[\P_{\mathbf{b}}(j,x) = \frac{1}{d'}\cos^2\left(\frac{\pi}{4} + (-1)^x \frac{\eps'b_j}{2}\right), \qquad \text{for all } j \in [d'], x \in \{0,1\}.\]
	Observe that with one call to $\fora_{\eps'}$, we can implement
	\begin{align*}
		&\frac{1}{\sqrt{2d'}} \sum_{j=1}^{d'} (\ket{j}\ket{0} + i\ket{j}\ket{1}) \overset{\fora_{\eps'}}{\mapsto} \frac{1}{\sqrt{2d'}} \sum_{j=1}^{d'} (\ket{j}\ket{0} + ie^{i\eps'b_j}\ket{j}\ket{1}) \\
		&\qquad = \sum_{j=1}^{d'} \frac{e^{i\frac{\pi}{4} + i\frac{\eps 'b_j}{2}}}{\sqrt{2d'}} \left(e^{-i\frac{\pi}{4} - i\frac{\eps'b_j}{2}}\ket{j}\ket{0} + e^{i\frac{\pi}{4} + i\frac{\eps'b_j}{2}}\ket{j}\ket{1}\right) \\
		&\qquad \overset{I \otimes (SH)}{\mapsto} \sum_{j = 1}^{d'} \frac{e^{i\frac{\pi}{4} + i\frac{\eps'b_j}{2}}}{\sqrt{d'}} \left(\cos\left(\frac{\pi}{4} + \frac{\eps'b_j}{2}\right)\ket{j}\ket{0} + \sin\left(\frac{\pi}{4} + \frac{\eps'b_j}{2}\right)\ket{j}\ket{1}\right) \\
		&\qquad = \sum_{(j,x) \in \Omega} \frac{e^{i\frac{\pi}{4} + i\frac{\eps'b_j}{2}}}{\sqrt{d'}} \cos\left(\frac{\pi}{4} + (-1)^x\frac{\eps'b_j}{2}\right)\ket{j}\ket{x} = \sum_{(j,x) \in \Omega} \sqrt{\P_{\mathbf{b}}(j,x)} e^{i\frac{\pi}{4} + i\frac{\eps'b_j}{2}} \ket{j}\ket{x},
	\end{align*}
	and hence we can implement $U_{\P_{\mathbf{b}}}$ with one call to $\mathcal{F}_{\eps'}$.\footnote{Note that we don't have to worry about the extra global phase here -- we can absorb it in the definition of the state $\ket{\omega}$, i.e., if $\omega = (j,x)$ we can define $\ket{\omega} = e^{i\pi/4 + i\eps'b_j/2} \ket{j}\ket{x}$, and then use all the machinery from the rest of this document.}

	Next, let the random variable $X : \Omega \to \R^d$ be defined as
	\[X(j,x) = \frac{x\sqrt{d'}}{4}H\mathbf{e}_j,\]
	where $H$ is a $(d' \times d')$-dimensional normalized Hadamard matrix, i.e., $H^TH = I$, whose first row and column have all positive signs. Such a Hadamard matrix exists because we know that $d'$ is a power of $2$. Furthermore,
	\[\mu = \ex{X} = \frac{1}{d'} \sum_{j=1}^{d'} \cos^2\left(\frac{\pi}{4} - \frac{\eps'b_j}{2}\right) \frac{\sqrt{d'}}{4}H\mathbf{e}_j = \frac18\mathbf{e}_1 + \frac{\sin(\eps')}{8\sqrt{d'}} H\mathbf{b} = \frac18\mathbf{e}_1 + \frac{\sqrt{d'}}{8n} H\mathbf{b},\]
	Thus, if we find an approximation $\widetilde{\mu}$ to $\mu$ such that $\norm{H(\widetilde{\mu} - \mu)}_1 \leq (d')^{3/2}/(16n)$, then
	\[\min_{\widetilde{\mathbf{b}} \in \{0,1\}^d} \norm*{\frac{8n}{\sqrt{d'}}\left(\widetilde{\mu} - \frac18\mathbf{e}_1\right) - H\widetilde{\mathbf{b}}}_1 = \frac{8n}{\sqrt{d'}} \min_{\widetilde{\mathbf{b}} \in \{0,1\}^d} \norm*{\widetilde{\mu} - \frac18\mathbf{e}_1 - \frac{\sqrt{d'}}{8n}H\widetilde{\mathbf{b}}}_1 \leq \frac{8n}{\sqrt{d'}} \norm{H(\widetilde{\mu} - \mu)}_1,\]
	and hence if we let $\widetilde{\mathbf{b}}$ be the bit string for which the minimum in the above expression is attained, then we find that
	\begin{align*}
		&\norm{H(\widetilde{\mathbf{b}} - \mathbf{b})}_1 \leq \norm*{\frac{8n}{\sqrt{d'}}\left(\widetilde{\mu} - \frac18\mathbf{e}_1\right) - H\widetilde{\mathbf{b}}}_1 + \norm*{\frac{8n}{\sqrt{d'}}\left(\widetilde{\mu} - \frac18\mathbf{e}_1\right) - H\mathbf{b}}_1 \\
		&\qquad \leq \frac{8n}{\sqrt{d'}}\norm{H(\widetilde{\mu} - \mu)}_1 + \frac{8n}{\sqrt{d'}}\norm{H(\widetilde{\mu} - \mu)}_1 = \frac{4n}{\sqrt{d'}} \norm{H(\widetilde{\mu} - \mu)}_1 \leq \frac{d'}{4}.
	\end{align*}
	We know that constructing such a bit string $\widetilde{\mathbf{b}}$ requires $\Omega(n)$ queries to $\fora_{\eps'}$, and hence we find that in order to find an $(d')^{3/2}/(16n)$-precise $\ell_1$-approximation of the mean of a random variable, we need to make at least $\Omega(n)$ calls to $U_{\P}$ as well. Thus, if we have only $n$ queries to spend, there must be an instance for which the $\ell_1$-approximation satisfies $\Omega(d^{3/2}/n)$. This completes the proof.
\end{proof}

We can now use the results obtained in Theorem~\ref{Thm:LBPOanalog} and Theorem~\ref{Thm:LBUPanalog} as black boxes to obtain results in different norms and regimes.

\begin{theorem}
  \label{Thm:LowerPO}
	Let $d \in \N$, $n,n' \geq 1$, and suppose that we have a quantum algorithm that finds an approximation $\widetilde{\mu}$ to the mean $\mu$ of any random variable with values contained in $[-1/4,1/4]^d$, using $n$ queries to $U_{\P}$ and $n'$ queries to $\pora_X$. Then, there exist instances such that
	\[\norm{\widetilde{\mu} - \mu}_1 = \begin{cases}
		\Omega\left(d\right), & \text{if } n' < d, \\[2mm]
		\Omega\left(\max\left\{\frac{d}{\sqrt{n}}, \frac{d^2}{n'}\right\}\right), & \text{if $n' \geq d$ and $n < d$,} \\[3mm]
		\Omega\left(\max\left\{\frac{d^{\frac32}}{n}, \frac{d^2}{n'}\right\}\right), & \text{if $n' \geq d$ and $n \geq d$,}
	\end{cases}\]
  with probability at least $2/3$. Moreover, the same expressions multiplied by $d^{1/p-1}$ can be obtained as lower bounds for $\norm{\widetilde{\mu} - \mu}_p$, for all $p \in (1,\infty]$.
\end{theorem}

\begin{proof}
	If $n,n' \geq d$, then we know from Theorem~\ref{Thm:LBPOanalog} and Theorem~\ref{Thm:LBUPanalog} that
	\[\norm{\widetilde{\mu} - \mu}_1 = \Omega\left(\max\left\{\frac{d^{\frac32}}{n}, \frac{d^2}{n'}\right\}\right).\]

	Next, let $1 \leq n < d$, and let $k = \lfloor d/n\rfloor$. Let $X' : \Omega \to \R^n$ be a random variable, and define $X : \Omega \to \R^d$ by $X(\omega) = X'(\omega) \otimes \mathbbm{1}_k$, padded with an appropriate number of zeros in the final entries. We find that
	\[\frac{1}{k} \sum_{\ell=1}^k \mu_{k(j-1)+\ell} = \frac1k \cdot k \cdot \mu_j = \mu_j,\]
	and hence, if we find an approximation of $\widetilde{\mu} \in \R^d$ to $\mu = \ex{X}$, then we can define $\widetilde{\mu}' \in \R^n$ as $\widetilde{\mu}'_j = \sum_{\ell=1}^k \widetilde{\mu}_{k(j-1)+\ell}/k$, which implies
	\begin{align*}
		&\norm{\widetilde{\mu}' - \mu'}_1 = \sum_{j=1}^n \left|\frac1k \sum_{\ell=1}^k \widetilde{\mu}_{k(j-1)+\ell} - \frac1k \sum_{\ell=1}^k \mu_{k(j-1)+\ell}\right| \leq \frac1k \sum_{j=1}^n \sum_{\ell=1}^k \left|\widetilde{\mu}_{k(j-1)+\ell} - \mu_{k(j-1)+\ell}\right| \\
		&\qquad = \frac1k\norm{\widetilde{\mu} - \mu}_1,
	\end{align*}
	and hence
	\[\norm{\widetilde{\mu} - \mu}_1 \geq k \norm{\widetilde{\mu}' - \mu'}_1 = \Omega\left(\frac{d}{n} \cdot \frac{n^{\frac32}}{n}\right) = \Omega\left(\frac{d}{\sqrt{n}}\right).\]

	Finally, the result for different values for $p \in (1,\infty]$ follows directly from H\"older's inequality, since for all $\mathbf{x} \in \R^d$, we have $\norm{\mathbf{x}}_p \geq d^{1/p-1}\norm{\mathbf{x}}_1$. This completes the proof.
\end{proof}

We aggregate all our results in Figure~\ref{Fig:PhaseOracle}.

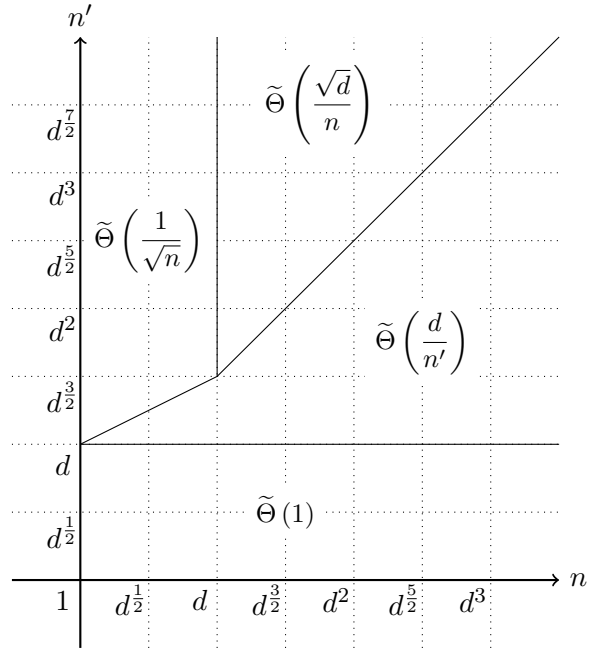
\begin{figure}[H]
	\centering
	\begin{tikzpicture}[scale=1.8]
		\draw[thick,->] (-.5,0) to (3.5,0) node[right] {$n$};
		\draw[thick,->] (0,-.5) to (0,4) node[above] {$n'$};
		\foreach \x in {.5,1,...,3} {
			\draw[dotted] (\x,-.5) to (\x,4);
		}
		\foreach \y in {.5,1,...,3.5} {
			\draw[dotted] (-.5,\y) to (3.5,\y);
		}
		\node[shift=({-.6em,-.7em})] at (0,0) {$1$};
		\node[shift=({-.6em,-.7em})] at (.5,0) {$d^{\frac12}$};
		\node[shift=({-.6em,-.7em})] at (1,0) {$d$};
		\node[shift=({-.6em,-.7em})] at (1.5,0) {$d^{\frac32}$};
		\node[shift=({-.6em,-.7em})] at (2,0) {$d^2$};
		\node[shift=({-.6em,-.7em})] at (2.5,0) {$d^{\frac52}$};
		\node[shift=({-.6em,-.7em})] at (3,0) {$d^3$};
		\node[shift=({-.6em,-.7em})] at (0,.5) {$d^{\frac12}$};
		\node[shift=({-.6em,-.7em})] at (0,1) {$d$};
		\node[shift=({-.6em,-.7em})] at (0,1.5) {$d^{\frac32}$};
		\node[shift=({-.6em,-.7em})] at (0,2) {$d^2$};
		\node[shift=({-.6em,-.7em})] at (0,2.5) {$d^{\frac52}$};
		\node[shift=({-.6em,-.7em})] at (0,3) {$d^3$};
		\node[shift=({-.6em,-.7em})] at (0,3.5) {$d^{\frac72}$};
		\draw (0,1) to (3.5,1);
		\draw (0,1) to (1,1.5);
		\draw (1,1.5) to (1,4);
		\draw (1,1.5) to (3.5,4);
		\node[fill=white] at (1.5,.5) {\small$\displaystyle\widetilde{\Theta}\left(1\right)$};
		\node[fill=white] at (2.5,1.75) {\small$\displaystyle\widetilde{\Theta}\left(\frac{d}{n'}\right)$};
		\node[fill=white] at (1.75,3.5) {\small$\displaystyle\widetilde{\Theta}\left(\frac{\sqrt{d}}{n}\right)$};
		\node[fill=white] at (.5,2.5) {\small$\displaystyle\widetilde{\Theta}\left(\frac{1}{\sqrt{n}}\right)$};
	\end{tikzpicture}
	\caption{Overview of the different regimes of the mean estimation problem. The horizontal and vertical axes show $n$ and $n'$, i.e., the number of queries to $U_{\P}$ and $\pora_X$, respectively. The complexities shown in the figure are the optimal error scaling of $\norm{\widetilde{\mu} - \mu}_{\infty}$ that can be achieved with particular choices for $n$ and $n'$. The tildes are hiding polylogarithmic factors in $n$, $n'$, $d$ and $1/\delta$. For the optimal error scalings of $\norm{\widetilde{\mu} - \mu}_p$ for $p \in [1,\infty)$, one can simply multiply the expressions above by $d^{1/p}$.}
  \label{Fig:PhaseOracle}
\end{figure}

%% file: Applications.tex
In this section, we describe some applications of our results. We first explain how our formulation of the multivariate mean estimation problem covers the general task of estimating the expectation values of several mutually commuting observables with respect to a given quantum state. We then present several applications in the literature, and notably in quantum machine learning, where this problem arises.

\subsection{Estimating expectation values of commuting observables}
\label{Sec:comObs}

\subsubsection{Classical versus quantum experiments}

From a classical perspective, the mean estimation problem is commonly described by a random experiment (or Monte Carlo process) that draws a \emph{classical} sample $\omega$ (e.g., a bit-string) from a certain probability space $(\Omega, 2^{\Omega}, \P)$ and leads to the associated observation $X(\omega)\in \R^d$ (see Definition~\ref{Def:rExp}). It is then quite clear that our generalization to quantum experiments, defined by a unitary $U_{\P}$ that prepares a superposition over basis states $\ket{\omega}$, $\omega \in \Omega$, and a unitary $\bora_X$ that evaluates $X$ for the same basis states (see Definitions~\ref{Def:qExp} and \ref{Def:binOracle}) can simulate such a random experiment. However, from a physical perspective, quantum experiments are also more general. The basis states $\{\ket{\omega}\}_{\omega \in \Omega}$ do not need to be computational basis states, and can be themselves superpositions of computational basis states or include arbitrary phases. Therefore, in our definition of quantum experiments, the unitaries $U_{\P}$ can indeed be arbitrary unitaries acting on a given Hilbert space $\Hil$, and what really matters here is the definition of the basis~$\{\ket{\omega}\}_{\omega \in \Omega}$ of $\Hil$.

\subsubsection{Problem definition}

With this observation, we can now move to the problem of estimating expectation values of mutually commuting observables. Let $U$ be a unitary transformation that prepares a given quantum state $\ket{\psi} = U\ket{0}$ in a given $m$-qubit Hilbert space~$\Hil$, and let $O_1, \ldots, O_d$ be $d$ \emph{mutually commuting} observables (i.e., Hermitian operators) acting on~$\Hil$. We want to compute estimates of the $d$ expectation values $\langle O_i \rangle = \bra{\psi}O_i\ket{\psi}$. Since the observables commute, they all share a common eigenbasis $\left\{\ket{\phi_j}\right\}_{1\leq j\leq 2^m}$, to which they assign eigenvalues $\lambda_{i} = (\lambda_{i,1}, \ldots, \lambda_{i,2^m}) \in \R^{2^m}$, respectively. Let us therefore look at the expression of $\ket{\psi} = U\ket{0}$ in this eigenbasis:
\begin{equation}
    U : \ket{0} \mapsto \sum_{j=1}^{2^m} \sqrt{\P(\phi_j)} e^{i\varphi_j} \ket{\phi_j}
\end{equation}
for some phases $\varphi_j \in [0,2\pi]$ and real amplitudes $\sqrt{\P(\phi_j)}$ such that $\sum_{j=1}^{2^m} \P(\phi_j) = 1$. If we now take $U_{\P}$ to be the unitary $U$, $\{\ket{\omega}\}_{\omega \in \Omega}$ to be $\left\{e^{i\varphi_j}\ket{\phi_j}\right\}_{1\leq j\leq 2^m}$, and $X(\omega)$ to be $X(\phi_j) = (\lambda_{1,j}, \ldots, \lambda_{d,j})$ (i.e., the eigenvalues $\lambda_{i,j}$ assigned by each of the observables $O_i$ to $\ket{\phi_j}$), the problem of estimating the $d$ expectation values $\langle O_i \rangle = \bra{\psi}O_i\ket{\psi}$ fits our formulation of the (quantum) mean estimation problem. Indeed, note that the phases $e^{i\varphi_j}$ do not contribute to the expectation values $\langle O_i \rangle = \bra{\psi}O_i\ket{\psi} = \sum_{j=1}^{2^m} \P(\phi_j) \lambda_{i,j}$, and therefore absorbing them in the basis states $\ket{\omega}$ does not influence the mean values to be estimated (nor our algorithms).

\subsubsection{Applicability assumptions}

For the applicability of our algorithms to this problem, we assume that a description of the eigenbasis $\left\{\ket{\phi_j}\right\}_{1\leq j\leq 2^m}$, in terms of a unitary transformation $V: \ket{j} \mapsto \ket{\phi_j}$ from computational basis states $\ket{j}$ to eigenvectors $\ket{\phi_j}$, and the eigenvalues $\{\lambda_{i}\}_{1\leq i\leq d}$ of the observables $\{O_i\}_{1\leq i\leq d}$ are known. These are the same assumptions that one would have in quantum algorithms for the univariate version of this problem (i.e., with one observable) \cite{KOS07j,WCNA09j} or in a setting where one would \emph{directly} estimate the expectation values $\langle O_i \rangle = \bra{\psi}O_i\ket{\psi} = \bra{\psi}V \Lambda_i V^{\dagger}\ket{\psi}$ for $\Lambda_i = \textrm{diag}(\lambda_{i,1}, \ldots, \lambda_{i,2^m})$ by applying $V^{\dagger}$ on $\ket{\psi}$, measuring computational basis states $\ket{j}$, and using several measurement outcomes $\{\ket{j}, (\lambda_{1,j}, \ldots, \lambda_{d,j})\}$ to simultaneously\footnote{When the observables do not commute, one cannot ``parallelize" measurements in such a manner, and would then be required to use different techniques like shadow tomography \cite{Aar20j,HKP20j}.} compute these estimates. Note that, in practice, the same transformation $V^{\dagger}$ would be absorbed in the unitary~$U_{\P}$ used in our algorithms, as to make the basis states $\{\ket{\omega}\}_{\omega \in \Omega}$ computational basis states, and ease the implementation of the unitaries $\bora_X$ and $\pora_X$ (using single-qubit rotations controlled by computational basis states).

\subsection{Examples of applications}
\label{Sec:ExApps}

\subsubsection{Training variational quantum circuits}
A straightforward application appears in some variational quantum algorithms for machine learning \cite{BLSF19j}. In a multidimensional regression setting \cite{MNKF18j} or a reinforcement learning setting \cite{JGM21p,SJD21p}, a variational quantum circuit defined by a parametrized and data-dependent unitary $U(x,\theta)$ and a set of observables $(O_1, \ldots, O_d)$ can be used as a hypothesis family $f_{\theta}(x) = (\langle O_1 \rangle_{x,\theta}, \ldots, \langle O_d \rangle_{x,\theta})$, for $\langle O_i \rangle_{x,\theta} = \bra{0^{\otimes n}} U^\dagger(x,\theta) O_i U(x,\theta)\ket{0^{\otimes n}}$, to model target functions $g$ with $d$-dimensional outputs. When the observables $O_1, \ldots, O_d$ all commute (e.g., commuting tensor products of Pauli operators or projectors on some basis states, for an arbitrary basis), the problem of estimating $f_{\theta}(x)$ fits the problem definition above.

\subsubsection{Training Boltzmann machines}
Another application considers the problem of estimating updates of a Boltzmann machine in a machine learning setting (e.g., a classification or generative modeling problem) \cite{WKS16j,WW19p,KW17j,JTN21j}. Take for instance a Boltzmann machine defined by a Hamiltonian:
\begin{equation}
H = \sum_{i<j} J_{i,j} \sigma^z_i\sigma^z_j + \sum_{i} b_i\sigma^z_i
\end{equation}
where $J_{i,j}$ and $b_i$ are real weights and biases and $\sigma^z_i$ is a Pauli-$Z$ operator acting on a qubit $i$ out of $n$ total qubits. The updates on the weights and biases of this Boltzmann machine take the form:
\begin{equation}
\Delta J_{i,j} = -\mathcal{L}(J,b)\langle\sigma^z_i\sigma^z_j\rangle, \quad \Delta b_{i} = -\mathcal{L}(J,b)\langle\sigma^z_i\rangle
\end{equation}
where $\mathcal{L}(J,b)$ is a loss dependent on the Boltzmann machine performance at the machine learning task and the expectation values $\langle\sigma^z_i\rangle, \langle\sigma^z_i\sigma_z^j\rangle$ are with respect to the Gibbs state:
\begin{equation}\label{eq:Gibbs-state}
  \ket{\psi} = \frac{1}{\textrm{Tr}_{x}[e^{-H}]} \sum_{x} \sqrt{e^{-H(x)}} \ket{x}
\end{equation}
for computational basis states $\ket{x}$. Assume having access to a unitary $U$ that prepares the Gibbs state of Equation~(\ref{eq:Gibbs-state}), e.g., using one of the subroutines in \cite{WKS16j,WW19p,KW17j,JTN21j}, then estimating the updates of the Boltzmann machine is an instance of the problem above for observables $\left\{ -\mathcal{L}(J,b)\sigma^z_i\sigma^z_j,\ -\mathcal{L}(J,b)\sigma^z_i\right\}_{i,j}$, i.e., weighted $\sigma^z_i$ and $\sigma^z_i\sigma^z_j$ operators, which are all diagonal in the computational basis.

\subsubsection{Training policies in reinforcement learning}

In the context of reinforcement learning \cite{SB18b}, an agent-environment interaction is described by a Monte Carlo process where, for a sequence of interactions, an agent acts probabilistically on its environment, the latter updates its state (probabilistically) depending on the actions of the agent and issues a real-valued reward $r_{t}$. The goal of the agent is to find a policy (i.e., a probability distribution $\pi(a_t|s_t)$ of actions $a_t$ given states $s_t$) that maximizes its \emph{expected} rewards $V(\pi) = \sum_{t=1}^T r_{t}$ for $T$ interactions with the environment. To do this, policy-based reinforcement learning algorithms define a certain family of parametrized policies $\pi_{\theta} \in \Pi_{\theta}$ (e.g., deep neural networks) and explore this policy family using gradient ascent on the expected rewards $V(\pi_{\theta})$. The so-called policy gradient theorem \cite{SMSM99c} gives a formulation of the gradient of the expected rewards $V(\pi_{\theta})$ with respect to the parameters $\theta \in \R^d$ of the policy as:
\begin{equation}
    \nabla_{\theta} V(\pi_{\theta}) = \underset{s_1, a_1, r_1, s_2, a_2, r_2, \ldots}{\mathbb{E}}\left[ \sum_{t=1}^T \nabla_{\theta} \log(\pi_{\theta}(a_t|s_t)) \sum_{t'=t}^T r_{t'}\right].
\end{equation}
This gradient is therefore given by the expectation value of the $d$-dimensional random variable
  \[X(s_1, a_1, r_1, s_2, \ldots) = \sum_{t=1}^T \nabla_{\theta} \log(\pi_{\theta}(a_t|s_t)) \sum_{t'=t}^T r_{t'}\]
(or equivalently, $d$ observables that are all diagonal in the computational basis) with respect to all possible interactions with the environment, following a policy $\pi_{\theta}$. In order for our mean estimators to be applicable here, our only assumption on the environment is that we have oracle access to its dynamics, notably its state-transitions \[\ket{s_t}\ket{a_t}\ket{0} \mapsto \sum_{s_{t+1}} \sqrt{P(s_{t+1}|s_t,a_t)}\ket{s_t}\ket{a_t}\ket{s_{t+1}}\] and its reward function \[\ket{s_t}\ket{0} \mapsto \ket{s_t}\ket{r_t}.\] As for the policy, we assume having the ability to implement $\pi_{\theta}$ coherently (i.e., similarly to $U_{\P}$), and to construct a (classical) circuit that computes the gradient $\nabla_{\theta} \log(\pi_{\theta}(a_t|s_t))$ given $s_t, a_t, \theta$.

%% file: Discussion.tex
In this work, we developed near-optimal quantum mean estimators in two different input models. In the binary oracle setting, we managed to obtain matching upper and lower bounds up to polylogarithmic factors, when we measure the performance of our estimator with respect to the Euclidean norm. We did not investigate the problem of deriving sharp bounds for other norms in this model. In the classical literature, sample-optimal estimators for general norms were given in~\cite{LM19ja}. One case that could be interesting to further study is the $\ell_{\infty}$-norm, since it arises naturally in our quantum algorithm as well as in the applications we consider. By combining the one-dimensional result with a union bound, one can obtain a classical estimator that achieves a precision of $\max_{j \in [d]} \sqrt{\var{X_j}\log(d/\delta)/n}$, whereas quantumly we obtained $\sqrt{\sum_{j \in [d]} \var{X_j}}\log(d/\delta)/n$. It would be interesting to figure out whether some combination of these two approaches can be shown to be optimal in all regimes.

One further observation is that we do not assume to have any knowledge about $\Sigma$ beforehand. Some preliminary considerations seem to indicate that in some $\ell_p$-norms, especially where $p < 2$, it might be useful to know bounds on the individual diagonal entries of this covariance matrix. Whether these considerations are fundamental, or can be worked around, is also an interesting question to address in the future.

%% file: Acknowledgements.tex
AC and SJ would like to thank Vedran Dunjko and M\={a}ris Ozols for pointing us in the direction of this problem, and for many insightful and motivating discussions. AC would also like to thank Ronald de Wolf and Joran van Apeldoorn for insightful discussions and helpful tips. Furthermore, AC would like to extend his gratitude to the anonymous legends that answered the Math Overflow post related to this research~\cite{MiscHamming21}. SJ acknowledges support from the Austrian Science Fund (FWF) through the projects DK-ALM:W1259-N27 and SFB BeyondC F7102. SJ also acknowledges the Austrian Academy of Sciences as a recipient of the DOC Fellowship.